\documentclass[12pt, draftclsnofoot, onecolumn]{IEEEtran}

\usepackage[utf8]{inputenc}
\usepackage[english]{babel}

\ifCLASSINFOpdf
\else
\fi

\usepackage{amsmath,amssymb,amsthm}
\usepackage{array}
\usepackage{graphicx}
\usepackage{mathtools}
\usepackage{xcolor}
\usepackage{subcaption}
\usepackage[linesnumbered,ruled,vlined]{algorithm2e}
\usepackage{algcompatible}

\usepackage{caption}

\DeclareUnicodeCharacter{2212}{-}

\DeclareMathOperator*{\argmax}{arg\,max}

\newcommand\blfootnote[1]{%
  \begingroup
  \renewcommand\thefootnote{}\footnote{#1}%
  \addtocounter{footnote}{-1}%
  \endgroup
}

\makeatletter

\newcommand{\nosemic}{\renewcommand{\@endalgocfline}{\relax}}
\newcommand{\dosemic}{\renewcommand{\@endalgocfline}{\algocf@endline}}
\let\oldnl\nl
\newcommand{\nonl}{\renewcommand{\nl}{\let\nl\oldnl}}
\makeatother

\newtheorem{thm}{Theorem}

\usepackage[pagewise]{lineno}
\begin{document}
\bstctlcite{IEEEexample:BSTcontrol}

\title{Affine Frequency Division Multiplexing for Next Generation Wireless Communications}

\author{Ali Bemani,~\IEEEmembership{Student Member,~IEEE,}
        Nassar Ksairi,~\IEEEmembership{Senior Member, IEEE,}\\
        and~Marios Kountouris,~\IEEEmembership{Senior Member,~IEEE.}
\thanks{A. Bemani and M. Kountouris are with the Communication Systems department, EURECOM, Sophia-Antipolis, France. Emails: ali.bemani@eurecom.fr, marios.kountouris@eurecom.fr. N. Ksairi is with the Mathematical and Algorithmic Sciences Lab, Huawei France R\&D, Paris, France. Email: nassar.ksairi@huawei.com}}

\maketitle

\vspace{-10mm}
\begin{abstract}
Affine Frequency Division Multiplexing (AFDM), a new chirp-based multicarrier waveform for high mobility communications, is introduced here. AFDM is based on discrete affine Fourier transform (DAFT), a generalization of discrete Fourier transform, which is characterized by two parameters that can be adapted to better cope with doubly dispersive channels. First, we derive the explicit input-output relation in the DAFT domain showing the effect of AFDM parameters in the input-output relation. Second, we show how the DAFT parameters underlying AFDM have to be set so that the resulting DAFT domain impulse response conveys a full delay-Doppler representation of the channel. Then, we show analytically that AFDM can achieve full diversity in doubly dispersive channels, where full diversity refers to the number of multipath components separable in either the delay or the Doppler domain, due to its full delay-Doppler representation. Furthermore, we present a low complexity detection method taking advantage of zero-padding. We also propose an embedded pilot-aided channel estimation scheme for AFDM, in which both channel estimation and data detection are performed within the same AFDM frame. Finally, simulations corroborate the validity of our analytical results and show the significant performance gains of AFDM over state-of-the-art multicarrier schemes in high mobility scenarios.
\end{abstract}

\IEEEpeerreviewmaketitle
\section{Introduction}
\blfootnote{
Part of the results appear in \cite{bemani2021afdm_ICC} and \cite{bemani2022low}.}
Next generation wireless systems and standards (beyond 5G/6G) are expected to support a wide spectrum of services, including reliable communication in high mobility scenarios (e.g., V2X communications, flying vehicles, and high-speed rail systems) and in extremely high-frequency (EHF) bands. Current systems are based on Orthogonal Frequency Division Multiplexing (OFDM), a widely used multicarrier scheme that achieves near optimal performance in time-invariant frequency selective channels. Nevertheless, in time-varying channels (also referred to as doubly dispersive or doubly selective channels), the performance of OFDM drastically decreases. This is mainly due to large Doppler frequency shifts and the loss of orthogonality among subcarriers, resulting in inter-carrier interference (ICI). This calls for new modulation techniques and waveforms, which are able to cope with various challenging requirements and show robustness in high mobility scenarios. 

One approach to compensate for fast variations in LTV channels is to shorten the OFDM symbol duration so that the channel variations over each symbol duration become negligible \cite{wang2006performance}. However, due to cyclic prefix (CP), this approach significantly reduces the spectral efficiency. In theory, the optimal approach to cope with time-varying multipath channels is to transmit information symbols leveraging an orthogonal eigenfunction decomposition of the channel and then project the received signal over the same set of orthogonal eigenfunctions at the receiver. In time-invariant (LTI) systems, complex exponentials are known to be eigenfunctions of the channel and can be obtained via the Fourier transform (FT). However, finding an orthonormal basis for general LTV channels is not trivial and polynomial phase models that generalize complex exponentials are often used as alternative bases. 

Since this optimal approach presents significant challenges both in terms of conceptual and computational complexity, using chirps, i.e., complex exponentials with linearly varying instantaneous frequencies, appears to be a promising alternative. The use of chirps for communication and sensing purposes has a long history. S. Darlington in 1947 proposed the chirp technique for pulsed radar systems with long-range performance and high-range resolution \cite{darlington1947}. 
The term “chirp” was apparently first employed by B. M Oliver in an internal Bell Laboratories Memorandum “Not with a bang, but a Chirp”.
In \cite{gott1971hf}, an experimental communication system employing chirp modulation in the high frequency band for air-ground communication is presented. Since chirped waveforms are of spread-spectrum, they can also provide security and robustness in several scenarios, including military, underwater and aerospace communications \cite{kadri2009low,he2010underwater,palmese2008experimental}. Chirps are specified in the IEEE 802.15.4a standard as chirp spreads spectrum (CSS)\pagebreak to meet the requirement of FCC on the radiation power spectral mask for the unlicensed UWB systems\cite{4299496}.

Using a frequency-varying basis for a multicarrier transmission scheme over time-varying channels is first introduced in \cite{martone2001multicarrier}. In this work,  an orthonormal basis formed by chirps are generated using Fractional Fourier Transform (FrFT). The scheme is presented in a continuous-time setting, whereas the approximation used for making the continuous-time FrFT discrete leads to imperfect orthogonality among chirp subcarriers and hence to performance degradation. A multicarrier technique based on Affine Fourier Transform (AFT), which is a generalization of the Fourier and  fractional Fourier transform, is proposed in \cite{erseghe2005multicarrier}. The resulting multicarrier waveform therein is referred to as DAFT-OFDM in the sequel where DAFT stands for Discrete AFT. It is equivalent to OFDM with reduced ICI on doubly dispersive channels and is shown to achieve low diversity order. Moreover, the delays and the Doppler shifts of channel paths are required at the transmitter in order to tune the DAFT-OFDM parameters. In \cite{stojanovic2010multicarrier}, a general interference analysis of the DAFT-OFDM system is provided and the optimal parameters are obtained in closed form, followed by the analysis of the effects of synchronization errors and the optimal symbol period. Another scheme that is proposed for communication over time-dispersive channels is Orthogonal Chirp Division Multiplexing (OCDM) \cite{ouyang2016orthogonal}, which is based on the discrete Fresnel transform - a special case of DAFT. OCDM is shown to perform better than uncoded OFDM in LTI and LTV channels \cite{omar2021performance}.
However, in LTI channels, OCDM can achieve unit diversity for very large signal to noise ratio (SNR), whereas in general LTV channels, it cannot achieve full diversity since its diversity order depends on the delay-Doppler profile of the channel.

In addition to chirp-based modulation, several waveforms have been proposed to provide improved performance compared to OFDM in terms of carrier frequency offset (CFO) sensitivity, peak to average power ratio (PAPR), and out-of-band emissions (OOBE). Discrete Fourier transform spread OFDM (DFT-s-OFDM), also known as Single Carrier-Frequency Division Multiple Access (SC-FDMA), has been proposed in \cite{myung2006single}, which spreads symbol energy equally over all subcarriers to reduce PAPR by precoding data symbols using a DFT. Generalized Frequency Division Multiplexing (GFDM) \cite{fettweis2009gfdm, michailow2014generalized} is a multicarrier modulation based on a circular pulse shaping filter, aiming to reduce the OOBE. Nevertheless, DFT-s-OFDM and GFDM are sensitive to CFO due to Doppler spread. \pagebreak To deal with the Doppler spread, Orthogonal Time Frequency Space (OTFS) modulation has recently been proposed for high mobility communications \cite{hadani2017orthogonal, hadani2018otfs}. OTFS is a two-dimensional (2D) modulation technique that spreads the information symbols over the delay-Doppler domain. OTFS has been shown to outperform previously proposed waveforms in both frequency selective and doubly selective channels \cite{anwar2020performance}. Therefore, in this paper, the performance of the proposed AFDM is compared with OTFS. OTFS achieves the full so-called \emph{effective} diversity order \cite{raviteja2019effective}, i.e., the diversity order in the finite SNR regime. However, \cite{surabhi2019diversity} shows that the OTFS diversity order without channel coding is one and with a phase rotation scheme using transcendental numbers can be made to achieve full diversity. The idea of embedding pilots along with the data symbols in the delay-Doppler domain has been proposed in \cite{raviteja2019embedded}. Although no separate transmission for the pilot symbols is needed, OTFS suffers from excessive pilot overhead due to its 2D structure as each pilot symbol should be separated from the data symbols. 

In this paper, we propose a novel multicarrier scheme called Affine Frequency Division Multiplexing (AFDM), which is a DAFT-based waveform using multiple orthogonal information-bearing chirp signals. 
The key idea is to multiplex information symbols in the DAFT domain in such a way that all the paths are separated from each other and each symbol experiences all paths coefficient. This separability is a unique feature of our scheme and cannot be achieved by other DAFT-based schemes. DAFT plays a fundamental role in AFDM, similarly to Fourier transform in OFDM. This work aims at establishing that AFDM is a promising new waveform for high mobility environments, having as well potential for communication at high frequency bands \cite{bemani2021affine}.
The contributions of this paper are summarized as follows:
\begin{itemize}
\item Introducing the Affine Fourier transform, we show how its discrete version can be achieved. Then, for the proposed AFDM, we analyze the DAFT domain input-output relation under doubly dispersive channels. 
The input-output relation is instrumental in giving insight on how DAFT parameters need to be tuned to avoid that time-domain channel paths with distinct delays or Doppler shifts overlap in the DAFT domain. 
\item We derive the diversity order of AFDM under maximum likelihood (ML) detection and we analytically show that AFDM achieves the full diversity of the LTV channels. 
\item We propose a low complexity detection algorithm for AFDM taking advantage of its inherent channel sparsity. For that, \pagebreak the channel matrix is approximated as a band matrix placing some null symbols - zero-padding the AFDM frame - in the DAFT domain. We propose a low complexity iterative decision feedback equalizer (DFE) based on weighted maximal ratio combining (MRC) of the channel impaired input symbols received from different paths. The overall complexity of this algorithm is linear both in the number of subcarriers and in the number of paths. We also show that this detector has similar performance as LMMSE detector with much less complexity.
\item For the embedded channel estimation, we arrange one pilot symbol and data symbols in one AFDM frame considering zero-padded symbols as guard intervals separating the pilot symbol and data symbols to avoid interference between them. We propose efficient approximated ML algorithms for channel estimation for the integer and fractional Doppler shifts. The proposed channel estimation schemes result in marginal performance degradation compared to AFDM with perfect channel knowledge.
\end{itemize}

This paper is organized as follows. In Section \ref{section_AFT}, AFT and DAFT are introduced and the proposed AFDM is presented in Section \ref{section_AFDM}. Diversity analysis of AFDM in LTV channels is provided in Section \ref{sec_diversity}. The proposed low complexity detection and channel estimation methods are presented in section \ref{sec_detection} and section \ref{section_channelEst}, respectively. In Section \ref{sec_sim_result}, simulation results for the AFDM performance are provided, and Section \ref{sec_conc} concludes this paper.

\section{Affine Fourier Transform}
\label{section_AFT}
In this section, we introduce the AFT and the DAFT, which form the basis of AFDM.
\subsection{Continuous Affine Frequency Transform}
Affine Fourier Transform, also known as Linear Canonical Transform \cite{healy2015linear}, is a four-parameter $(a,b,c,d)$ class of linear integral transform defined as
\begin{equation}
    S_{a,b,c,d}(u) = \begin{cases} \int_{-\infty}^{+\infty}s(t)K_{a,b,c,d}(t,u){\rm d}t,& b\neq 0 \\s(df){e^{-\imath{cd\over 2}u^{2}}\over \sqrt{a}}, \qquad \qquad \qquad \quad \,& b=0\end{cases}\label{AFT}
\end{equation}
where $(a,b,c,d)$ forms $\textrm{M} = \left[{\begin{array}{cc}
a & b \\ c & d
\end{array}}\right]$ with unit determinant, i.e, $ad - bc = 1$ and transform kernel given by
\begin{equation}
    K_{a,b,c,d}(t, u) = \frac{1}{\sqrt{2\pi\vert b\vert }}{{\rm e}^{-\imath({a\over 2b}u^{2}+{1\over b}ut+{d\over 2b}t^{2})}}.
\end{equation}
The inverse transform can be expressed as an AFT having the parameters $\textrm{M}^{-1} = \left[{\begin{array}{cc}
a & -b \\ -c & d
\end{array}}\right]$ 
\begin{equation}
    s(t) = \int_{-\infty}^{+\infty}S_{a,b,c,d}(u)K_{a,b,c,d}^*(t,u){\rm d}u.
\end{equation}
The AFT generalizes several known mathematical transforms, such as Fourier transform (0,1/2$\pi$,-2$\pi$,0), Laplace transform (0,$j(1/2\pi)$,$j2\pi$,0), $\theta$-order fractional Fourier transform  ($\cos$$\theta$, (1/2$\pi$)$\sin$$\theta$,-2$\pi\sin\theta$,$\cos\theta$), Fresnel transform and the scaling operations. The extra degree of freedom of AFT provides flexibility and has been employed in many applications, including filter design, time-frequency analysis, phase retrievals, and multiplexing in communication. The effect of AFT can be interpreted by the Wigner distribution function (WDF). After doing the AFT, the WDF of $S_{a,b,c,d}(u)$ will be the twisting of the WDF of $s(t)$.

\subsection{Discrete Affine Frequency Transform}
The discrete transform can generally be used either to compute the continuous transform for spectral analysis or to process discrete data signals. Sampling the continuous function provides the input of the discrete transform in the former case, while a pure discrete sequence is considered for the input in the latter case. 
Therefore, discrete AFT is obtained in two types \cite{pei2000closed}, which are essentially identical with different parameterizations. To derive the DAFT, input function $s(t)$ and $S_{a,b,c,d}(u)$ are sampled by the interval $\Delta t$ and $\Delta u$ as 
\begin{equation}
    s_n = s(n\Delta t), S_m = S_{a,b,c,d}(m\Delta u),
    \label{eq:y_Y_discrete}
\end{equation}
where $n = 0, ..., N-1$ and $m = 0, ..., M-1$. From \eqref{eq:y_Y_discrete}, we can convert \eqref{AFT} as
\begin{equation}
    S_m = \frac{1}{\sqrt{2\pi\vert b\vert }}\cdot \Delta t \cdot {{\rm e}^{-\imath({a\over 2b}m^2\Delta u^{2})}\sum_{n = 0}^{N-1}e^{-\imath({1\over b}mn\Delta u \Delta t+{d\over 2b}n^2\Delta t^{2})}}s_n.
\end{equation}
This equation can be written in the form of transformation matrix 
\begin{equation}
    S_m = \sum_{n = 0}^{N-1}F_{a,b,c,d}(m, n)s_n,\label{eq:tranMat}
\end{equation}
where $F_{a,b,c,d}(m, n) = \frac{1}{\sqrt{2\pi\vert b\vert }} \cdot\Delta t \cdot{{\rm e}^{-\imath({a\over 2b}m^2\Delta u^{2}+{1\over b}mn\Delta u \Delta t+{d\over 2b}n^2\Delta t^{2})}}$.

In order for \eqref{eq:tranMat} to be reversible, the following condition should hold \cite{pei2000closed}
\begin{equation}
    \Delta t \Delta u = \frac{2\pi|b|}{M}.\label{eq:condrivers}
\end{equation}
Thus, the DAFT of the first type can be written as follows:
\begin{align}
    &S_m = \frac{1}{\sqrt{M}}{{\rm e}^{-\imath{a\over 2b}m^2\Delta u^{2}}\sum_{n = 0}^{N-1}e^{-\imath({2\pi\over M}mn+{d\over 2b}n^2\Delta t^{2})}}s_n, \quad b >0\\
    &S_m = \frac{1}{\sqrt{M}}{{\rm e}^{-\imath{a\over 2b}m^2\Delta u^{2}}\sum_{n = 0}^{N-1}e^{-\imath(-{2\pi\over M}mn+{d\over 2b}n^2\Delta t^{2})}}s_n, \quad b <0.
\end{align}
The DAFT of the second type \cite{pei2000closed} can be obtained by defining $c_1 = \frac{d}{4\pi b}\Delta t^2$ and $c_2 = \frac{a}{4\pi b}\Delta u^2$ so that $S_m$ in \eqref{eq:tranMat} writes as $S_m=\sum_{n=0}^{N-1}F_{c_1,c_2}(m,n)s_n$, where
\begin{equation}
    F_{c_1,c_2}(n, m) \triangleq \frac{1}{\sqrt{M}}{{\rm e}^{-\imath2\pi( c_2m^2 +{sgn(b)\over M}mn+c_1n^2)}}.
\end{equation}
The condition in \eqref{eq:condrivers} then becomes $c_1c_2 = \frac{ad}{4M^2}$.
Since $a$ and $d$ can take any real value as long as $b$ and $c$ are adjusted to satisfy $ad-bc=1$, there is no constraint for $c_1$ and $c_2$ and can take any real values. Further simplification follows from fixing $sgn(b)=1$, i.e., the DAFT is defined as
\begin{equation}
    S_m = \frac{1}{\sqrt{M}}{{\rm e}^{-\imath2\pi c_2m^2}\sum_{n = 0}^{N-1}e^{-\imath2\pi({1\over M}mn+c_1n^2)}}s_n,\label{eq:Y_m_DAFT}
\end{equation}
where $M\geq N$ and its inverse transform is the following
\begin{equation}
    s_n = \frac{1}{\sqrt{M}}{{\rm e}^{\imath2\pi c_1n^2}\sum_{m = 0}^{M-1}e^{\imath2\pi({1\over M}mn+c_2m^2)}}S_m.\label{eq:y_n_DAFT}
\end{equation}

Moreover, we should take into account that sampling in one domain imposes periodicity in another domain. Considering \eqref{eq:Y_m_DAFT} and \eqref{eq:y_n_DAFT}, the following periodicity can be seen
\begin{align}
    &S_{m+kM} =e^{-\imath2\pi c_2(k^2M^2 + 2kMm)}S_m,\\
    &s_{n+kN} =e^{\imath2\pi c_1(k^2N^2 + 2kNn)}s_n.\label{eq:perTime}
\end{align}
For our purposes, only constraint \eqref{eq:perTime} matters, whose sole practical effect is on the kind of prefix one should add to a DAFT-based multicarrier symbol. When $M = N$, as considered in this paper, the inverse transform is the same as the forward transform with parameters $-c_1$ and $-c_2$ and conjugating the Fourier transform term. In matrix representation, arranging samples $s_n$ and $S_m$ in the period $[0, N)$ in vectors
\begin{equation}
\begin{aligned}
    \mathbf{s} &= (s_0, s_1, ..., s_{N-1}) \ \quad \ \text{and} \ \quad \ \mathbf{S} &= (S_0, S_1, ..., S_{N-1}),
    \label{eq:chirp_periodicity}
\end{aligned}
\end{equation}
DAFT is expressed as $\mathbf{S} = \mathbf{As}$ with $\mathbf{A} = \mathbf {\Lambda}_{c_2}{\mathbf F}{\mathbf \Lambda}_{c_1}$, $\mathbf{F}$ being the DFT matrix with entries $e^{-\imath2\pi mn/N}/\sqrt{N}$ and 
\begin{equation}
    \mathbf{\Lambda}_{c}= {\rm diag} (e^{-\imath2\pi cn^{2}}, n=0, 1, \, \ldots\,, N-1).
\end{equation}
The inverse of the matrix $\mathbf{A}$ is given by $\mathbf{A}^{-1} = \mathbf{A}^H =  {\mathbf \Lambda}_{c_1}^H{\mathbf F}^H{\mathbf \Lambda}_{c_2}^H$.
We can now show that $F_{c_1,c_2}(n, m)$ with $sgn(b) = 1$ and $M = N$ forms an orthonormal basis of $\mathbb{C}^N$, i.e,
\begin{equation}
\begin{aligned}
    & \sum_{n = 0}^{N-1}F_{c_1,c_2}(n, m_1)F^*_{c_1,c_2}(n, m_2) = \frac{1}{N}e^{-\imath2\pi c_2(m_1^2-m_2^2)}\sum_{n = 0}^{N-1}e^{-\imath\frac{2\pi}{N}(m_1 - m_2)n} = \delta(m_1 - m_2).
\end{aligned}
\end{equation}

\vspace{-2mm}
\section{Affine Frequency Division Multiplexing}\label{section_AFDM}
In this section, we present our proposed DAFT-based multicarrier waveform and transceiver scheme, coined AFDM. In this scheme, inverse DAFT (IDAFT) is used to map data symbols into the time domain, while DAFT is performed at the receiver to obtain the effective discrete affine Fourier domain channel response to the transmitted data, as shown in Fig.~\ref{fig:AFDM_blokcdiagrma}.
\vspace{-5mm}
\subsection{Modulation}
Let $\mathbf{x} \in \mathbb{A}^{N\times 1}$ denote the vector of information symbols in the discrete affine Fourier domain, where $\mathbb{A}$
$\subset \mathbb{Z}[j]$ represents the alphabet and $\mathbb{Z}[j]$ denotes the number field whose elements have the form $z_r + z_ij$, with $z_r$ and $z_i$ integers. QAM symbols are considered in the remainder. The modulated signal can be written as 
\begin{equation}
\label{eq:mod}
    s[n] = \sum_{m = 0}^{N-1}x[m]\phi_n(m), \quad n= 0 , \cdots, N-1,
\end{equation}
where $\phi_n(m) = \frac{1}{\sqrt{N}}\cdot e^{\imath2\pi (c_1n^2 + c_2m^2 + nm/N)}$. In matrix form, \eqref{eq:mod} becomes ${\mathbf s} = \mathbf{A}^H\mathbf{x}={\mathbf \Lambda}_{c_1}^H{\mathbf F}^H{\mathbf \Lambda}_{c_2}^H\mathbf{x}$.

\begin{figure*}
  \centering
  \includegraphics[scale=.6]{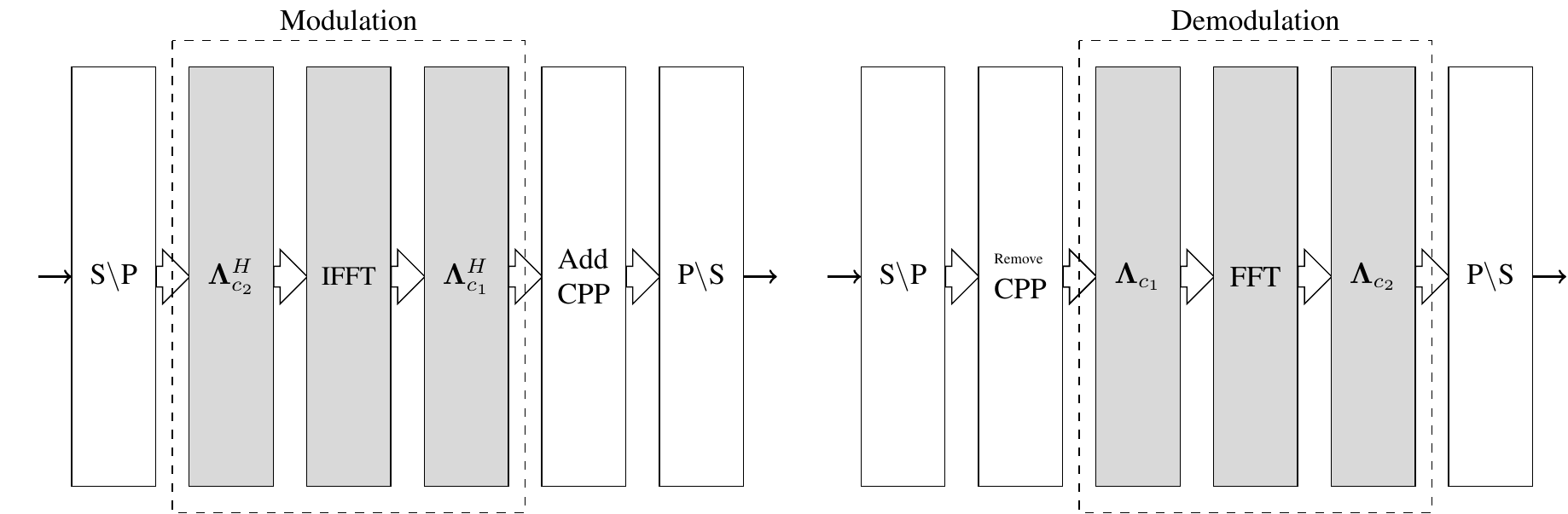}
  \caption{AFDM block diagram}
  \label{fig:AFDM_blokcdiagrma}
  \vspace{-5mm}
\end{figure*}

Similarly to OFDM, the proposed scheme requires a prefix to combat multipath propagation and make the channel seemingly lie in a periodic domain. Due to different signal periodicity, a {\emph{chirp-periodic}} prefix (CPP) is used here instead of an OFDM cyclic prefix (CP). For that, an $L_{\rm{cp}}$-long prefix, occupying the positions of the negative-index time-domain samples, should be transmitted, where $L_{\rm{cp}}$ is any integer greater than or equal to the value in samples of the maximum delay spread of the channel. With the periodicity defined in \eqref{eq:chirp_periodicity}, the prefix is
\begin{equation}
    s[n] = s[N+n]e^{-\imath2\pi c_1(N^2+2Nn)},\quad n = -L_{\rm{cp}}, \cdots, -1.
\end{equation} 
Note that a CPP is simply a CP whenever $2Nc_1$ is an integer value and $N$ is even.
\subsection{Channel}
\label{sec:channel}
After parallel to serial conversion and transmission over the channel, the received samples are 
\begin{equation}
    r[n] = \sum_{l = 0}^{\infty}s[n-l]g_n(l) + w[n],\label{eq:r_n1}
\end{equation}
where $w_n\sim\mathcal{CN}\left(0,N_0\right)$ is an additive Gaussian noise and
\begin{equation}
    g_n(l) = \sum _{i=1}^{P} h_{i}e^{-\imath2\pi f_in}\delta(l - l_i),\label{eq:channel_model}
\end{equation}
is the impulse response of channel at time $n$ and delay $l$, where $P\geq1$ is the number of paths, $\delta(\cdot)$ is the Dirac delta function, and $h_i, f_i$ and $l_i$ are the complex gain, Doppler shift (in digital frequencies), and the integer delay associated with the $i$-th path, respectively. Note that this model is general and also covers the case where each delay tap can have a Doppler frequency \emph{spread} by simply allowing for different paths $i,j\in\{1,\ldots,P\}$ to have the same delay $l_i=l_j$, while satisfying $f_i\neq f_j$.
We define $\nu_i \triangleq{} N f_i = \alpha_i +a_i $, where $\nu_i\in\left[-\nu_{\max},\nu_{\max}\right]$ is the Doppler shift normalized with respect to the subcarrier spacing, $\alpha_i\in\left[-\alpha_{\max},\alpha_{\max}\right]$ is its integer part whereas $a_i$ is  the fractional part satisfying $\frac{-1}{2} < a_i \leq \frac{1}{2}$. We assume that the maximum delay of the channel satisfies $l_{\max}\triangleq\max (l_i) < N$, and that the CPP length is greater than $l_{\max}-1$. 

After discarding the CPP, we can write \eqref{eq:r_n1} in the matrix form
\begin{equation}
    {\mathbf r} = {\mathbf H}\mathbf{s} + \mathbf{w},
\end{equation}
where $\mathbf{w}\sim\mathcal{CN}\left(\mathbf{0},N_0\mathbf{I}\right)$ and ${\mathbf H}= \sum _{i=1}^{P} h_{i} {\mathbf {\Gamma }}_{\mathrm{CPP}_{i}} {\mathbf {\Delta }}_{f_{i}} {\mathbf{\Pi }}^{l_{i}}$,
where $\mathbf{\Pi }$ is the   forward   cyclic-shift   matrix, $
\Delta_{f_i} \triangleq{\mathrm{diag}}(e^{-\imath2\pi f_in}, n = 0, 1, \, \ldots,\, N-1)
$
and ${\mathbf {\Gamma }}_{\mathrm{CPP}_{i}} $ is a $N\times N$ diagonal matrix 
{{
\begin{equation}
    \begin{multlined}
    \mathbf {\Gamma }_{\mathrm{CPP}_{i}} =  {\mathrm{diag}}(\begin{cases}e^{-\imath2\pi c_1(N^{2}-2N(l_i - n))}&n < l_i\\ 1&n\geq l_i\end{cases}, n = 0,\, \ldots,\, N-1).
    \end{multlined}
    \label{eq:CP_matrix}
\end{equation}
}}
From \eqref{eq:CP_matrix} we can see that whenever $2Nc_1$ is an integer and $N$ is even, $\mathbf {\Gamma }_{\mathrm{CPP}_{i}}=\mathbf {I}$.
\vspace{-2mm}
\subsection{Demodulation}
At the receiver side, the DAFT domain output symbols are obtained by
\begin{equation}
    y[m] = \sum_{n = 0}^{N-1}r[n]\phi_n^*(m).
    \label{eq:y_received}
\end{equation}
In matrix representation, the output can be written as
\begin{equation}
\begin{aligned}
    {\mathbf y} =& \mathbf{A}\mathbf{r}
    = \sum _{i=1}^{P} h_{i} \mathbf{A}{\mathbf {\Gamma }}_{\mathrm{CPP}_{i}} {\mathbf {\Delta }}_{f_{i}} {\mathbf{\Pi }}^{l_{i}} \mathbf{A}^H{\mathbf x}+ \mathbf{A}{\mathbf w}
    = {\mathbf H}_{\mathrm{eff}} {\mathbf x}+ \widetilde {\mathbf w},\label{eq_rec}
\end{aligned}
\end{equation}
where  ${\mathbf H}_{\mathrm{eff}} \triangleq \mathbf{A}{\mathbf H}\mathbf{A}^H$ and $ \widetilde {\mathbf w} = \mathbf{A}\mathbf{w}$. Since $\mathbf{A}$ is a unitary matrix, $\widetilde {\mathbf w}$ and ${\mathbf w}$ have the same statistical properties.

\subsection{Input-Output relation }\label{subsec:input_output_rel}
From \eqref{eq_rec}, we see that the received symbols are a linear combination of the transmitted symbols. Moreover, we know that features, such as diversity order, detection complexity, and channel estimation, are determined by the input-output relation, i.e, the structure of the effective channel. For example, the OFDM effective channel is diagonal, exhibiting poor diversity while the detection can be implemented using a 1-tap equalizer. For that, we provide here the structure of ${\mathbf H}_{\mathrm{eff}}$ as input-output relation and show that it has a sparse structure and can be formed by the AFDM parameters.
Considering the definition of ${\mathbf H}_{\mathrm{eff}}$, \eqref{eq_rec} can be rewritten as
\begin{equation}
    \mathbf{y} = \sum_{i = 1}^{P}h_i\mathbf{H}_i\mathbf{x} + \Tilde{\mathbf{w}}, \label{eq:y_output}
\end{equation}
where $\mathbf{H}_i \triangleq {\mathbf A}{\mathbf {\Gamma }}_{\mathrm{CPP}_{i}} {\mathbf {\Delta }}_{f_{i}} {\mathbf{\Pi }}^{l_{i}}{\mathbf A}^H\label{H_i}$.
It can be shown that ${H}_i[p, q]$ is given by
\begin{equation}
\begin{aligned}
    &{H}_i[p, q] = \frac{1}{N}e^{\imath\frac{2\pi}{N}(Nc_1l_i^2 - ql_i + Nc_2(q^2 - p^2))}\mathcal{F}_i(p, q),\label{eq:H_i_general}
\end{aligned}
\end{equation}
 where we denote $\mathcal{F}_i(p, q)$ as
\begin{equation}
\begin{aligned}
\mathcal{F}_i(p, q) &= \sum_{n = 0}^{N-1} e^{-\imath\frac{2\pi}{N}((p-q + \nu_i + 2Nc_1 l_i)n)} = \frac {e^{-\imath {2\pi } (p-q + \nu_i + 2Nc_1 l_i) }\!-\!1}{e^{-\imath \frac {2\pi }{N} (p-q + \nu_i + 2Nc_1 l_i)}\!-\!1}.\label{eq:general_F_eq}
\end{aligned}
\end{equation}
As we can see, the value of $\mathcal{F}_i(p, q)$ depends on the Doppler shift $\nu_i$. Therefore, we have two cases, namely integer Doppler shift and fractional Doppler shift. We first show the input-output relation for the integer case, and we state the relation of the general case afterwards.
\subsubsection{Integer Doppler Shifts}
With $\nu_i$ being integer valued for all $i\in\{1,\ldots,P\}$, i.e., $a_i=0$, \eqref{eq:general_F_eq} is equal to
\begin{equation}
    \mathcal{F}_i(p, q) = \begin{cases}
    N & q = (p + \mathrm{loc}_i)_N\\
    0 & {\text{otherwise}}
    \end{cases}
\end{equation}
where $\mathrm{loc}_i \triangleq (\alpha_i + 2Nc_1 l_i)_N$, $(\cdot)_N$ is the modulo $N$ operation and \eqref{eq:H_i_general} writes as
\begin{equation}
    {H}_i[p, q] = 
    \begin{cases} e^{\imath\frac{2\pi}{N}(Nc_1l_i^2 - ql_i + Nc_2(q^2 - p^2))} & q = (p+\mathrm{loc}_i)_N  \\
    0 & {\text{otherwise}}
    \end{cases}.
    \label{eq:Hi_p_q_integer}
\end{equation}
Hence, there is only one non-zero element in each row of $\mathbf{H}_i$ as shown in Fig.~\ref{fig:integerDopplerOnePath}, and the input-output relation for \eqref{eq:y_output} becomes
\begin{equation}
    y[p] = \sum_{i = 1}^{P}h_ie^{\imath\frac{2\pi}{N}(Nc_1{l_i^2} - ql_i + Nc_2(q^2 - p^2))}x[q] + \Tilde{{w}}[p],  0 \leq p \leq N-1,
    \label{eq:y_output_special}
\end{equation}
where $q = (p + \mathrm{loc}_i)_N$.

\subsubsection{Fractional Doppler Shifts}
Considering the fractional Doppler shifts, it can be shown that for a given $p$, $\mathcal{F}_i(p, q) \neq 0$, for all $q$. However, the magnitude of ${H}_i[p, q]$ has a peak at $q = (p + \mathrm{loc}_i)_N$ and decreases as $q$ moves away from $\mathrm{loc}_i$. To show this, we have
\begin{equation}
\begin{aligned}
    \left|{H}_i[p, q]\right| &= \left| \frac{e^{\imath\frac{2\pi}{N}(Nc_1{l_i^2} - ql_i + Nc_2(q^2 - p^2))}}{N}\mathcal{F}_i(p, q)\right| = \left| \frac{1}{N}\mathcal{F}_i(p, q)\right| = \left|\frac {\sin(N\theta)}{N\sin(\theta)}\right|,
    \end{aligned}
\end{equation}
 where $\theta \triangleq \frac{\pi}{N}(p-q + \mathrm{loc}_i+a_i)$. Using the inequality $ \left|\sin(N\theta) \right|\leq \left|N\sin(\theta)\right|$, which is tight for small values of $\theta$, we can show 
 \begin{equation} 
 \begin{aligned}
 \left |{\frac {\sin (N\theta)}{N\sin (\theta)}}\right |&=\left |{\frac {\sin ((N\!-\!1)\theta)\cos (\theta) + \sin (\theta)\cos ((N\!-\!1)\theta) }{N\sin (\theta)}}\right |\le\frac {N-1}{N}\left |{\cos (\theta)}\right |+\frac {1}{N}.\label{eq:sin_theta}
 \end{aligned}
 \end{equation}

 The right-hand side (r.h.s.) of \eqref{eq:sin_theta} has its peak at the smallest $|\theta_{p,q,i}|$ when $q=(p+loc_i)_N$. As $q$ moves away from $(p+loc_i)_N$, $|\theta_{p,q,i}|$ increases and the r.h.s. of \eqref{eq:sin_theta} decreases. Moreover, the larger the value of $N$, the faster this decrease is with respect to $q$. Therefore, we consider from now on that $|H_i(p,q)|$ is non-zero only for $2k_\nu+1$ values of $q$ corresponding to an interval centered at $ q=(p+loc_i)_N $. Here, $k_\nu$ is chosen as a function of $N$ in such a way that for all $i$ and $p$ the r.h.s. of \eqref{eq:sin_theta} is smaller than a sensitivity threshold for all $ |q-(p+loc_i)_N|>k_\nu$. In formulas, this translates to
 \begin{equation}
    {H}_i[p, q] =
    \begin{cases} \frac{e^{\imath\frac{2\pi}{N}(Nc_1{l_i^2} - ql_i + Nc_2(q^2 - p^2))}}{N}\mathcal{F}_i(p, q) & (p + \mathrm{loc}_i - k_\nu)_N \leq q \leq (p + \mathrm{loc}_i + k_\nu)_N\\
    0 & {\text{otherwise}}
    \end{cases}.
    \label{eq:Hi_p_q_fractional}
\end{equation}
Hence, there are $2k_\nu + 1$  non-zero elements in each row of $\mathbf{H}_i$, as shown in Fig.~\ref{fig:FractionalDopplerOnePath}, and the input-output relation for \eqref{eq:y_output} is written as
{\small
\begin{equation}
    y[p] = \sum_{i = 1}^{P}\frac{1}{N}h_ie^{\imath2\pi c_1{l_i^2}}\sum_{q = (p + \mathrm{loc}_i - k_\nu)_N}^{(p + \mathrm{loc}_i + k_\nu)_N}e^{\imath\frac{2\pi}{N}(-ql_i + Nc_2(q^2 - p^2))}\frac {e^{\imath {2\pi } (p-q + a_i + \mathrm{loc}_i) }\!-\!1}{e^{\imath \frac {2\pi }{N} (p-q +a_i + \mathrm{loc}_i)}\!-\!1}x[q] + \Tilde{{w}}[p],  0 \leq p \leq N-1.
    \label{eq:y_output_special_fractional}
\end{equation}
}
It should be noted that the range for the sum is circulant, i.e, when it is from $N-3$ to $1$, it is counted as $N-3, N-2, N-1, 0, 1$.
\begin{figure}
\centering
\begin{subfigure}{0.3\textwidth}
\centering
\includegraphics[width= \textwidth]{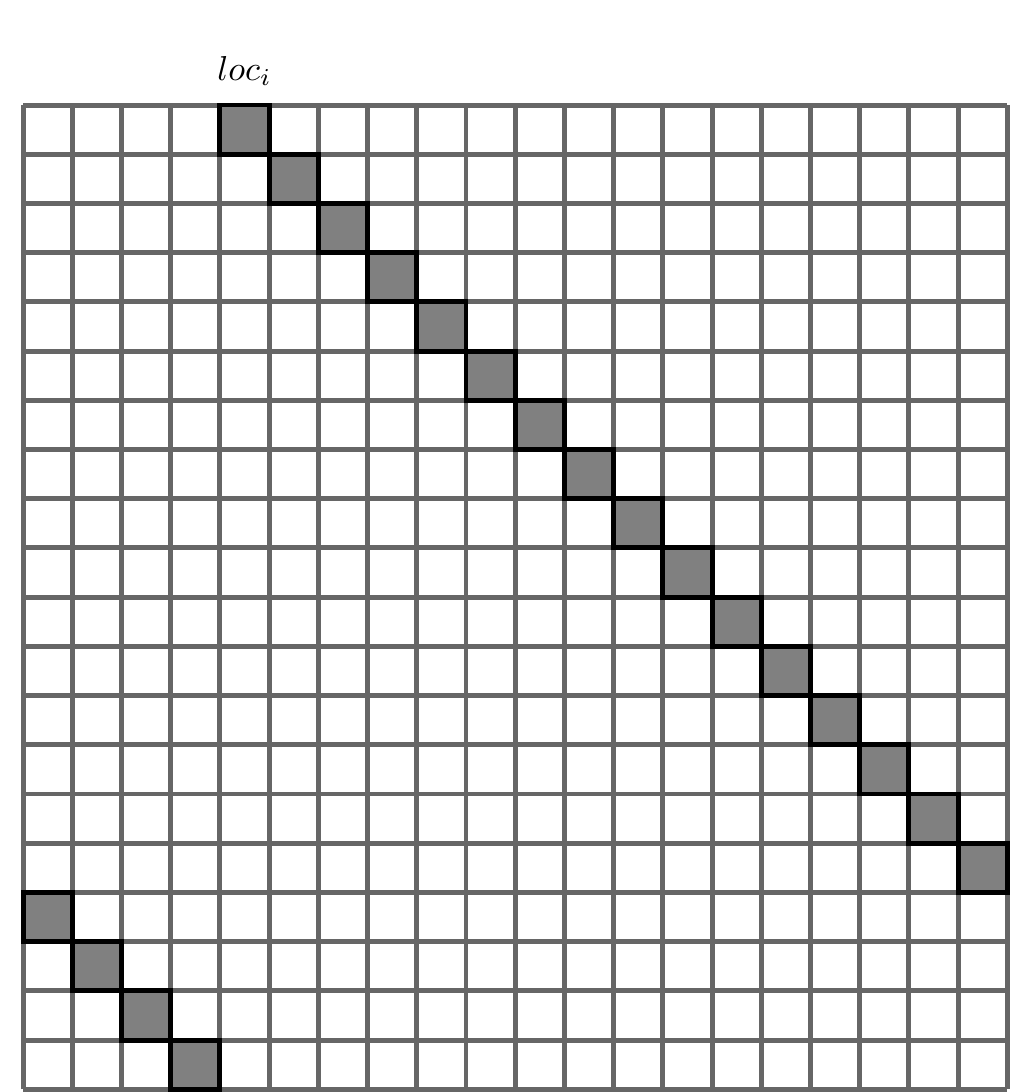}
\caption{Integer Doppler shift}
\label{fig:integerDopplerOnePath}
\end{subfigure}
\hspace{0.1\textwidth}
\begin{subfigure}{0.3\textwidth}
\centering
\includegraphics[width= \textwidth]{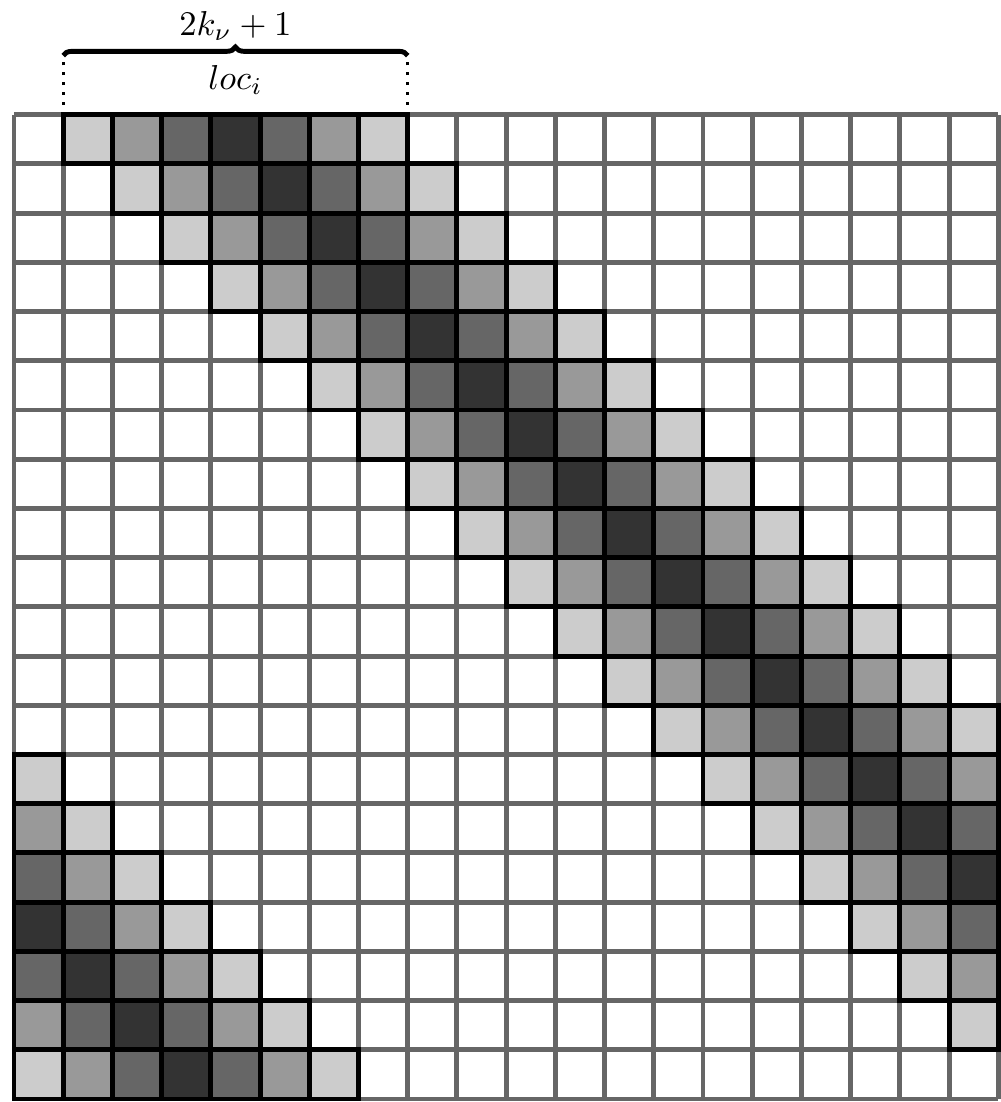}
\caption{Fractional Doppler shift}
\label{fig:FractionalDopplerOnePath}
\end{subfigure}
\caption{Structure of $\mathbf{H}_i$ for (a) integer and (b) fractional Doppler shifts.} 
\label{fig:OnePathStructure}
\end{figure}

\vspace{-6mm}
\subsection{AFDM Parameters}
The performance of DAFT-based modulation schemes critically depends on the choice of parameters $c_1$ and $c_2$. For instance, OCDM uses $c_1 = c_2 = \frac{1}{2N}$, whereas in DAFT-OFDM $c_2 = 0$ and $c_1$ are adapted to the delay-Doppler channel profile to minimize ICI. Nevertheless, both schemes fail to achieve full diversity in LTV channels, as shown in Section~\ref{sec_sim_result}. In the proposed AFDM, we find $c_1$ and $c_2$ for which the DAFT domain impulse response constitutes a full delay-Doppler representation of the channel. This allows AFDM to achieve full diversity in LTV channels, as shown in Section \ref{sec_diversity}.
For that, the non-zero entries in each row of $\mathbf{H}_i$ for each path $i\in\{1,\ldots,P\}$ should not coincide with the position of the non-zero entries of the same row of $\mathbf{H}_j$ for any $j\in\{1,\ldots,P\}$ such that $j\neq i$. Observing \eqref{eq:Hi_p_q_integer} and \eqref{eq:Hi_p_q_fractional}, the location of each path depends on its delay-Doppler information and AFDM parameters. For the integer and fractional Doppler shifts, $\mathrm{loc}_i$ and $\mathrm{loc}_{i, frac}$ are in the following range
\begin{equation}
    -\alpha_{\rm{max}} + 2Nc_1l_i \leq \mathrm{loc}_i \leq 
    \alpha_{\rm{max}} + 2Nc_1l_i,
    \label{eq:range_integer}
\end{equation}
\begin{equation}
    -\alpha_{\rm{max}}-k_\nu + 2Nc_1l_i  \leq \mathrm{loc}_{i, frac} \leq 
    \alpha_{\rm{max}} + k_\nu + 2Nc_1l_i ,
    \label{eq:range_fractional}
\end{equation}
respectively. Therefore, for the positions of the non-zero entries of $\mathbf{H}_i$ and $\mathbf{H}_j$ to not overlap, the intersection of the corresponding ranges of $\mathrm{loc}_i$ and $\mathrm{loc}_j$ for the integer case and $\mathrm{loc}_{i, frac}$ and $\mathrm{loc}_{j, frac}$ for the fractional case should be empty, i.e, 
\begin{equation}
\begin{aligned}
    \{-\alpha_{\rm{max}} + 2Nc_1l_i, ..., \alpha_{\rm{max}} + 2Nc_1l_i\}\cap
    \{-\alpha_{\rm{max}} + 2Nc_1l_j, ..., \alpha_{\rm{max}} + 2Nc_1l_j\}= \emptyset,
    \label{eq:intersec_int}
\end{aligned}
\end{equation}
{\small
\begin{equation}
\begin{aligned}
    \{-\alpha_{\rm{max}} - k_\nu +2Nc_1l_i, ..., \alpha_{\rm{max}} +k_\nu + 2Nc_1l_i\}\cap
    \{-\alpha_{\rm{max}}-k_\nu+ 2Nc_1l_j, ..., \alpha_{\rm{max}}+k_\nu + 2Nc_1l_j\}= \emptyset.
    \label{eq:intersec_frac}
    \end{aligned}
\end{equation}}
For the paths with different delays, assuming $l_i < l_j$, satisfying \eqref{eq:intersec_int} and \eqref{eq:intersec_frac} is equivalent to satisfying the constraints
\begin{align}
    &2Nc_1 > \frac{2\alpha_{\rm{max}}}{l_j - l_i},\\
    &2Nc_1 > \frac{2(\alpha_{\rm{max}} + k_\nu)}{l_j - l_i}.
\end{align}

If there is no sparsity in the time-domain impulse response of the channel, then the minimum value of $l_j - l_i $ is one, and for the integer case, $c_1$ is set to
\begin{equation}
\label{eq:opt_c1_int}
    {c_1 = \frac{2\alpha_{\max} + 1}{2N},}
\end{equation}
and for the fractional case, $c_1$ is chosen as
\begin{equation}
\label{eq:opt_c1_frac}
    {c_1 = \frac{2(\alpha_{\max} + \xi_\nu) + 1}{2N}},
\end{equation}
for some $\xi_\nu\leq k_\nu$. Through $\xi_\nu$, there is flexibility in setting $c_1$ and reducing pilot overhead (see section \ref{section_channelEst}) at the expense of $|\mathbf{H}_{\rm eff}|$ no longer being strictly circulant. Moreover, the only remaining condition for the DAFT-domain impulse response to constitute a full delay-Doppler representation of the channel is to ensure that the non-zero entries of any two matrices $\mathbf{H}_{i_{\min}}$ and $\mathbf{H}_{i_{\max}}$ corresponding to paths $i_{\min}$ and $i_{\max}$ with delays $l_{i_{\min}}\triangleq\min_{i=1\cdots P}l_i$ and $l_{i_{\max}}\triangleq\max_{i=1\cdots P}l_i$, respectively, do not overlap due to the modular operation in \eqref{eq:Hi_p_q_integer} and \eqref{eq:Hi_p_q_fractional}. This overlapping never occurs if $2\alpha_{\rm{max}}l_{\rm{max}} + 2\alpha_{\rm{max}} + l_{\rm{max}}< N$ for the integer case and $2(\alpha_{\rm{max}} + k_\nu)l_{\rm{max}} + 2(\alpha_{\rm{max}}+k_\nu) + l_{\rm{max}}< N$ for the fractional case. Since wireless channels are usually underspread (i.e., $l_{\rm{max}} \ll N$ and $\alpha_{\rm{max}} \ll N$), this condition is satisfied even for moderate values of $N$. 

 With this parameter setting, channel paths with different delay values or different Doppler frequency shifts become separated in the DAFT domain, resulting in $\textbf{H}_{\mathrm{eff}}$ having the structure shown in Fig.~\ref{fig:Channelpic} (for the fractional case, $\alpha_{\max}$ should be replaced with $\alpha_{\max} + \xi_\nu$). Thus, we get a delay-Doppler representation of the channel in the DAFT domain since the delay-Doppler profile can be determined from the positions of the non-zero entries in any row of $\textbf{H}_{\mathrm{eff}}$. This feature can be obtained by neither DAFT-OFDM (since by construction the effective channel matrix is made as close to diagonal as possible to reduce ICI), nor OCDM (since setting $c_1 =\frac{1}{2N}$, there may exist two paths $i\neq j$ such that the non-zero entries of $\mathbf{H}_i$ and $\mathbf{H}_j$ coincide under some delay-Doppler profiles of the channel). In the next section, we show that this unique feature of AFDM translates into being diversity order optimal in LTV channels. 
\begin{figure}
  \centering
  \includegraphics[scale=.6]{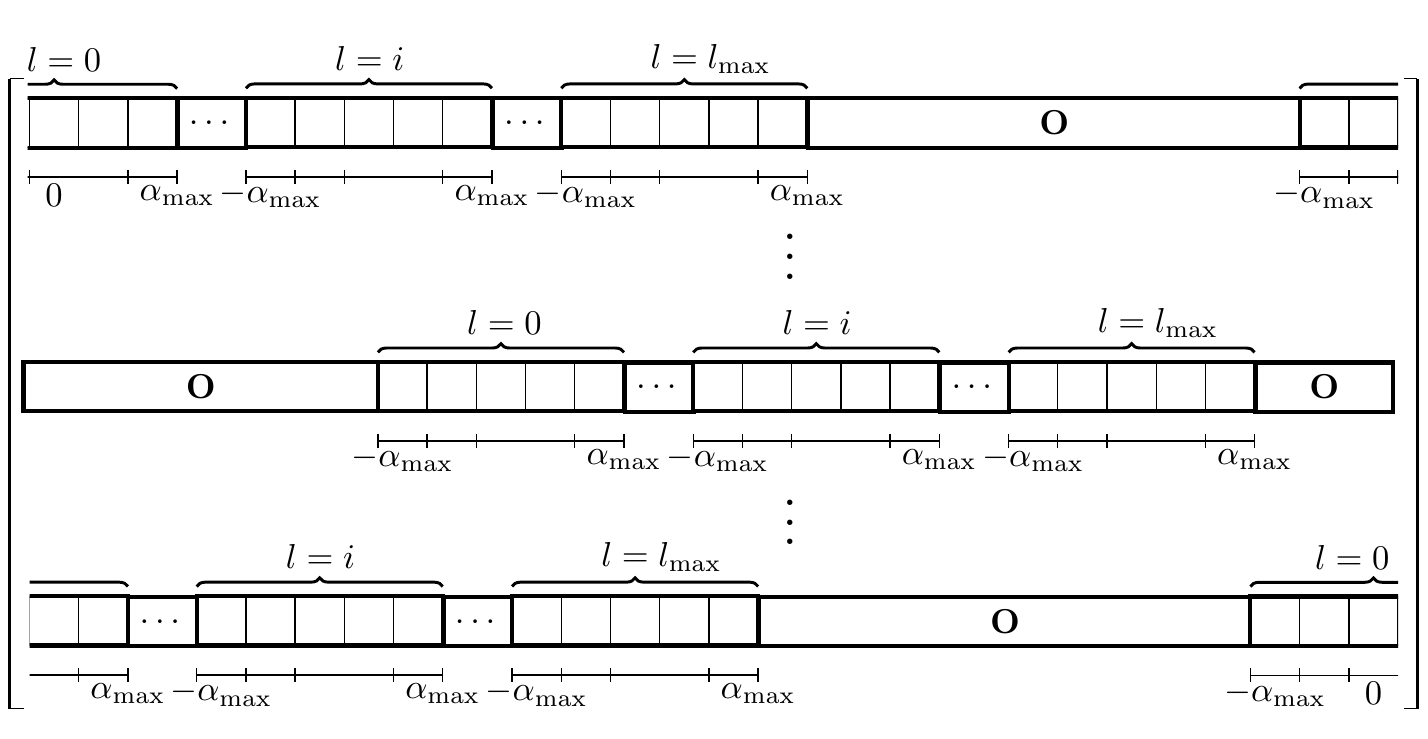}
  \caption{Structure of $\textbf{H}_{\mathrm{eff}}$ in AFDM for the integer case.}
  \label{fig:Channelpic}
\end{figure}

\section{Diversity Analysis}\label{sec_diversity}
In this section, we analyze the diversity order of AFDM. Due to space limitations, diversity analysis is presented only in the case of integer Doppler shifts. However, Theorem \ref{theo:main} given below, also holds for the fractional Doppler case.
To this end, we rewrite \eqref{eq_rec} as 
\begin{equation}
{\mathbf y}= \sum _{i=1}^{P} h_{i} {\mathbf H}_{i}{\mathbf x}+ \widetilde {\mathbf w} = {\mathbf{\Phi }}({\mathbf {x}}) {\mathbf h} + \widetilde {\mathbf w},\label{eq:diff_paths}
\end{equation}
where $\mathbf{h}=[h_1,h_2, \ldots,h_P]$ is a $P\times 1$ vector and ${\boldsymbol{\Phi}}({\mathbf{x}})$ is the $N \times P$ concatenated matrix 
\begin{equation}
    {\boldsymbol{\Phi}}({\mathbf x}) = [\mathbf{H}_1\mathbf{x} \mid \ldots \mid \mathbf{H}_P\mathbf{x}].
    \label{eq:phi_x}
\end{equation}
We now express the pairwise error probability (PEP). First, we normalize the elements of $\mathbf {x}$ so that the average energy of $\mathbf {x}$ is one, thus the signal-to-noise ratio (SNR) is given by $\frac{1}{N_0}$. Assuming perfect channel state information and ML detection at the receiver, the conditional PEP between $\mathbf{x}_m$ and $\mathbf{x}_n$, i.e., transmitting symbol $\mathbf{x}_m$ and deciding in favor of $\mathbf{x}_n$ at the receiver, can be upper bounded as
\begin{equation}
P(\mathbf {x}_{m}\rightarrow \mathbf {x}_{n}) \leq \prod \limits _{l=1}^{r}\frac {1}{1+\,\,\dfrac {\lambda _{l}^{2}}{4PN_0}}\label{upper_bound}
\end{equation}
where  $\lambda_l$ is the $l$-th singular value of the matrix $\mathbf{\Phi}(\boldsymbol{\delta}^{(m, n)})$ and $\boldsymbol{\delta}^{(m, n)} = \mathbf{x}_m - \mathbf{x}_n$.
At high SNR, \eqref{upper_bound} becomes
\begin{equation}
P(\mathbf {x}_{m}\rightarrow \mathbf {x}_{n}) \leq \frac {1}{\frac{1}{{N_0}^{r}}\prod \limits _{l=1}^{r} \dfrac {\lambda _{l}^{2}}{4P}}.\label{BER}
\end{equation}
We can see from \eqref{BER} that the exponent of the SNR term, $\frac{1}{N_0}$, is r, which is equal to the rank of the matrix $\boldsymbol{\Phi}(\boldsymbol{\delta}^{(m, n)})$. The overall bit error rate (BER) is dominated by the PEP with the minimum value of r, for all $m, n, m\neq n$. Hence, the diversity of AFDM is given by
\begin{equation} 
\rho = \min _{m,n ~m\neq n}~\text {rank}(\boldsymbol{\Phi }(\boldsymbol{\delta}^{(m, n)})) \leq P.
\end{equation}

For exposition convenience, we drop $(m, n)$ from $\boldsymbol{\delta}^{(m, n)}$. 
First, we show that for DAFT-based multi-carrier schemes, a necessary (but not sufficient) condition to achieve full diversity, i.e., $\rho=P$, is
\begin{equation}
\label{eq:cond_non_colocated}
\forall  i,j\in [1, \cdots, P],\quad
\mathrm{loc}_i \neq \mathrm{loc}_j.
\end{equation}
The above condition can hold for AFDM, but not for OCDM and DAFT-OFDM. For that sake, we assume that there exist channel paths $i$ and $j$ such that the locations of their corresponding non-zero elements in the channel matrix are the same, i.e., $\mathrm{loc}_i = \mathrm{loc}_j$. We then show that full diversity cannot be achieved under this assumption. Indeed, for full diversity, the following matrix composed of the columns $i$ and $j$ of $\boldsymbol{\Phi}(\boldsymbol{\delta})$ should be of rank 2 for all possible values of $\boldsymbol{\delta}$
{\small
\begin{equation}
\label{eq:mat_i_j}
\left[{\begin{array}{cc}{H}_i[1, q_0]{\delta}[{q_0}]& {H}_j[1, q_0]{\delta}[{q_0}] \\ {H}_i[2, (q_0+1)_N]{\delta}[{(q_0+1)_N}] & {H}_j[2, (q_0+1)_N]{\delta}[{(q_0+1)_N}] \\ \vdots & \vdots \\ {H}_i[N, (q_0+N-1)_N]{\delta}[{(q_0+N-1)_N}] & {H}_j[N, (q_0+N-1)_N]{\delta}[{(q_0+N-1)_N}       ]
\end{array}}\right],
\end{equation}
}
where $q_0 = \mathrm{loc}_i = \mathrm{loc}_j$. However, when $\boldsymbol{\delta}$ is such that it has a single non-zero element, the two columns of the matrix in \eqref{eq:mat_i_j} are dependent, hence the rank of this matrix cannot be 2. Consequently, the rank of $\boldsymbol{\Phi}(\boldsymbol{\delta})$ cannot be $P$ and the condition in \eqref{eq:cond_non_colocated} is thus necessary for full diversity. Therefore, proving that AFDM achieves full diversity is equivalent to proving that tuning $c_2$ can make matrix $\boldsymbol{\Phi}(\boldsymbol{\delta})$ to be full rank. 

\begin{thm}\label{theo:main}
For a linear time-varying channel with a maximum delay $l_{\rm{max}}$ and maximum normalized Doppler shift $\alpha_{\rm{max}}$, AFDM with $c_1$ satisfying \eqref{eq:opt_c1_int} achieves full diversity ($\rho=P$) if 
\begin{align}
    2\alpha_{\rm{max}} + l_{\rm{max}} + 2\alpha_{\rm{max}}l_{\rm{max}} < N. \label{eq:cond_1}
\end{align}     
\end{thm}
\begin{proof}
See Appendix \ref{sec_appen_proof}.
\end{proof}
\vspace{-5mm}
\section{Low-Complexity Weighted MRC-based DFE Detection}\label{sec_detection}
Although for showing the diversity order of AFDM ML detection is used, it is prohibitively complex to implement in real-world communication systems. For that, in this section, we propose a low-complexity detector. The first step is to place some null symbols that allow approximating the truncated part of $\textbf{H}_{\rm{eff}}$ as a band matrix. This also simplifies the input-output relation as the modular operation is no longer needed (see Fig.~\ref{fig:ChannelForMMSE_withoutpilot}). Note that these symbols do not entail extra overhead as they can serve not only the proposed detection algorithms but also embedded pilot aided channel estimation. Due to the structure of ${\bf H}_{\rm eff}$, the number of the null guard symbols should be greater than
\begin{equation}
Q \triangleq (l_{\max} + 1)(2(\alpha_{\max} + \xi_\nu) + 1)-1.\label{eq:Q_val}
\end{equation}
Taking into account the zero padding, the vector of DAFT domain received samples writes as
\begin{equation}
    {\mathbf y} = \underline{\mathbf H}_{\mathrm{eff}} \underline{\mathbf x}+ \widetilde{\underline{\mathbf w}},\label{eq:detectionEq}
\end{equation}
where $\underline{\mathbf x}$ and $\underline{\mathbf H}_{\mathrm{eff}}$ are the truncated parts of ${\mathbf x}$ and $\mathbf{H}_{\mathrm{eff}}$, respectively (see Fig.~\ref{fig:ChannelForMMSE_withoutpilot}). They can be expressed using the matrix $\mathbf{T} = [\mathbf{I}_N]_{Q-(\alpha_{\max}+\xi_\nu):N-(\alpha_{\max}+\xi_\nu)-1, :}$ as $\underline{\mathbf x} = \mathbf{T}\mathbf{x}$ and  $\underline{\mathbf H}_{\mathrm{eff}} = {\mathbf H}_{\mathrm{eff}} \mathbf{T}^{{H}}$.
\begin{figure}
\centering
\begin{minipage}{0.37\textwidth}
  \centering
  \includegraphics[width= \textwidth]{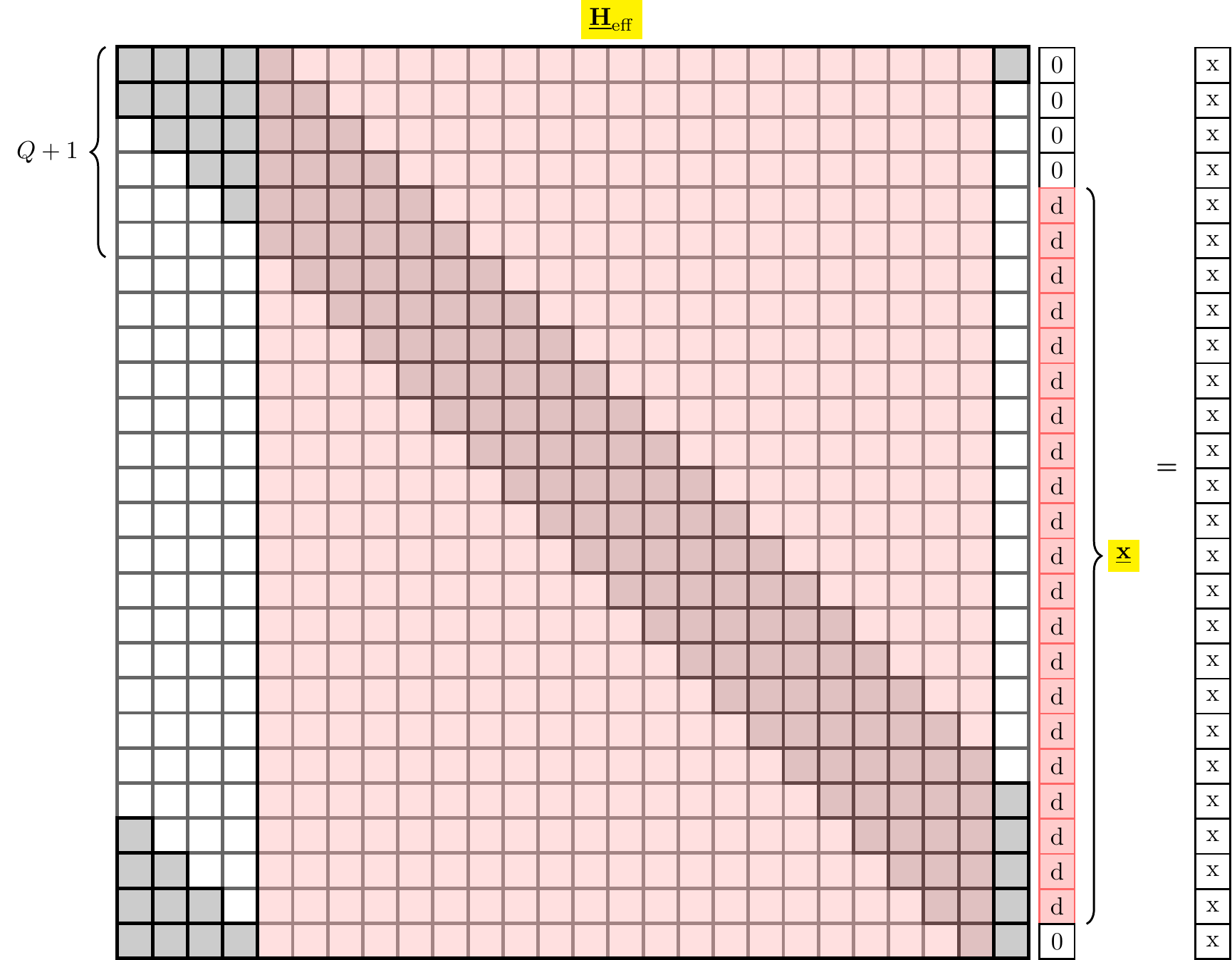}
  \captionof{figure}{Truncated parts of ${\mathbf x}$ and $\mathbf{H}_{\mathrm{eff}}$}
  \label{fig:ChannelForMMSE_withoutpilot}
\end{minipage}
\hspace{0.1\textwidth}
\begin{minipage}{0.5\textwidth}
\centering
 \includegraphics[width= \textwidth]{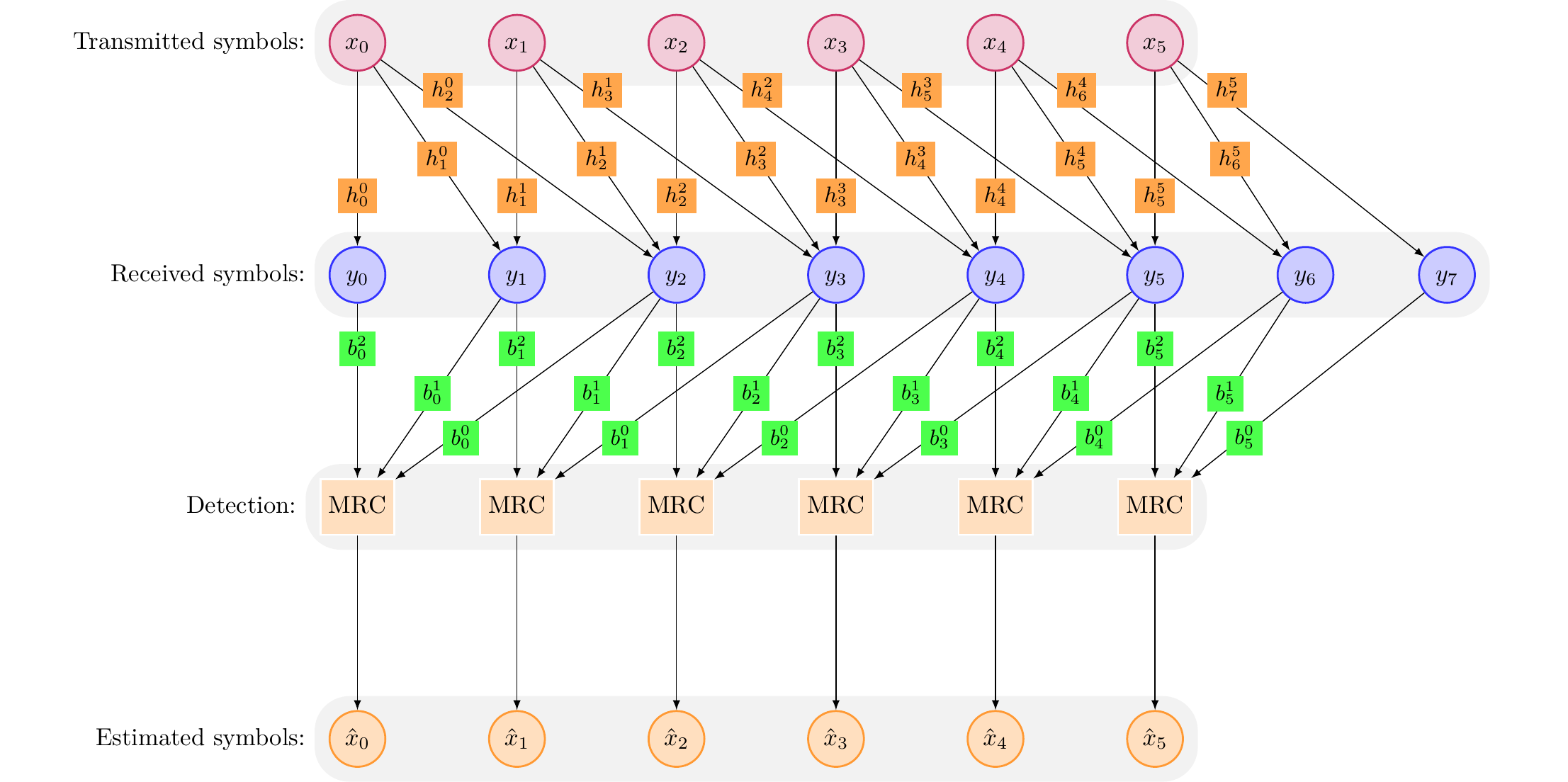}
  \captionof{figure}{Weighted MRC operation for $N = 8$ with a 3-path channel with $Q=2$.}
  \label{fig:MRC_detector}
\end{minipage}
\end{figure}
Using LMMSE equalization for \eqref{eq:detectionEq} requires $\mathcal{O}(N^3)$ flops, which can be prohibitive for large $N$. We hence propose a weighted MRC-based DFE exploiting the sparse representation of the communication channel provided by AFDM.

As shown in Fig.~\ref{fig:ChannelForMMSE_withoutpilot},
$\underline{\mathbf H}_{\mathrm{eff}}$ has $L$ non-zero entries per column, where $L = P$ for the integer Doppler shift case and $L= (2\xi_\nu + 1)P$ for fractional Doppler shift case, with $L\leq Q$ in both cases. Each of these non-zero entries of a column of $\underline{\mathbf H}_{\mathrm{eff}}$ provides a copy of the data symbol corresponding to the index of this column. We propose a detection scheme where each data symbol is detected from a weighted maximum ratio combining (MRC) of its $L$ channel-impaired received copies. Fig.~\ref{fig:MRC_detector} depicts an instance of this detector for AFDM with $N = 8$ and a $3$-path channel with $Q=2$. The proposed detector is iterative, wherein each iteration, the estimated inter symbol interference is canceled in the branches selected for the combining. Considering the structure of $\underline{\mathbf H}_{\mathrm{eff}}$, it can be seen that each received symbol $y[k]$ is given by
\vspace{-4mm}

\begin{equation}
    y[k] = \sum_{i = 0}^{L-1}\underline{ H}_{\mathrm{eff}}[k, {p}_k^i]x[{p}_k^i],
\end{equation}
where ${p}_k^i$ is the column index of the $i$-th path coefficient in row $k$ of matrix $\underline{\mathbf H}_{\mathrm{eff}}$. Let $b_k^i$ be the channel impaired input symbol $x[k]$ in the received samples $y[{{q}_k^i}]$ after canceling the interference from other input symbols, where ${q}_k^i$ is the row index of the $i$-th path coefficient in column $k$ of matrix $\underline{\mathbf H}_{\mathrm{eff}}$. In each iteration, assuming estimates of the input symbols $x[k]$ are available, either from the current iteration (for ${p}_{{q}_k^i}^j < k, \quad j = 0, ..., L-1$) or previous iteration (for $ {p}_{{q}_k^i}^j> k, \quad j = 0, ..., L-1$), $b_k^i$ can be written as 
\begin{equation}
    b_k^i = y[{{q}_k^i}] - \sum_{ {p}_{{q}_k^i}^j< k} \underline{ H}_{\mathrm{eff}}[{{{q}_k^i}}, { {p}_{{q}_k^i}^j}]\hat{x}[{ {p}_{{q}_k^i}^j}]^{(n)}- \sum_{ {p}_{{q}_k^i}^j> k} \underline{H}_{\mathrm{eff}}[{{q}_k^i}, { {p}_{{q}_k^i}^j}]\hat{x}[{ {p}_{{q}_k^i}^j}]^{(n-1)},\label{eq:b_k_i}
\end{equation}
where superscript $(n)$ denotes the $n$-th iteration. It can be seen that for each symbol $x[k]$, we need to compute $L$ scalars. This operation has complexity order of $\mathcal{O}(L^2)$. However, when computing $b_k^i$ for all symbols $k$, there are some redundant operations involved that can be avoided by instead computing $b_k^i$ as follows
\vspace{-3mm}
\begin{equation}
    b_k^i = \Delta y[{{q}_k^i}]^{(n)} + \underline{H}_{\mathrm{eff}}[{q}_k^i,k]\hat{x}[{k}]^{(n - 1)}.\label{eq:new_bki}
\end{equation}
Here, $\Delta y_{q_k(i)}^{(n)}$ is the residual error remaining while reconstructing the received symbols and is given by
\begin{equation}
\vspace{-2mm}
    \Delta y[{{q}_k^i}]^{(n)} = y[{{q}_k^i}] - \sum_{ {p}_{{q}_k^i}^j< k} \underline{ H}_{\mathrm{eff}}[{{{q}_k^i}}, { {p}_{{q}_k^i}^j}]\hat{x}[{ {p}_{{q}_k^i}^j}]^{(n)}- \sum_{ {p}_{{q}_k^i}^j\geq k} \underline{H}_{\mathrm{eff}}[{{q}_k^i}, { {p}_{{q}_k^i}^j}]\hat{x}[{ {p}_{{q}_k^i}^j}]^{(n-1)}.\label{eq:delta_y}
    \vspace{-2mm}
\end{equation}
Define $g_k^{(n)}$ and $d$ as
\begin{equation}
     g_k^{(n)} \triangleq \sum_{i = 0}^{L-1}\underline{ H}_{\mathrm{eff}}^{*}[{q}_k^i, k]b_k^i = \sum_{i = 0}^{L-1}\underline{ H}_{\mathrm{eff}}^{*}[{q}_k^i, k]\Delta y[{q}_k^i]^{(n)} + dx[k]^{(n - 1)},
\end{equation}
\begin{equation}
    d \triangleq \sum_{i = 0}^{L-1}|\underline{ H}_{\mathrm{eff}}[{q}_k^i, k]|^2.
\end{equation}
It should be noted that since $|{\mathbf H}_{\mathrm{eff}}|$ is a circulant matrix, the value of $d$ is independent of $k$ and needs to be computed only once.
We now denote the SNR by $\gamma$. Instead of directly using $g_k^{(n)}/d$ as the estimate of $x_k$ (which would have amounted to using the MRC criterion), we define the symbol estimate as
\begin{equation}
\hat{x}[k]^{(n)} = c_k^{(n)},
\end{equation}
\begin{equation}
    c_k^{(n)}  \triangleq \frac{g_k^{(n)}}{d + \gamma^{-1}} = \frac{1}{d + \gamma^{-1}}\sum_{i = 0}^{L-1}\underline{ H}_{\mathrm{eff}}^{*}[{q}_k^i, k]\Delta y[{ {q}_k^i}]^{(n)} + \frac{d}{d + \gamma^{-1}}\hat{x}[k]^{(n-1)}.
\end{equation}
We later show (see Section \ref{subsec:rel_GS}) that this weighting of $g_k^{(n)}$ while computing $\hat{x}_k^{(n)}$ guarantees that the iterative detection algorithm converges to the LMMSE estimate of the symbols vector $\mathbf{\underline{x}}$.
In each iteration, after estimation of each symbol $x[k]^{(n)}$, the values of $\Delta y[{{q}_k^i}]^{(n)}$ for $i=0,…,L-1$ need to be updated using
\begin{equation}
    \Delta y[{{q}_k^i}]^{(n)} = \Delta y[{{q}_k^i}]^{(n)} - \underline{ H}_{\mathrm{eff}}^{*}[{q}_k^i, k]x[k]^{(n)} - x[k]^{(n-1)}).
\end{equation}
Once all symbols are estimated, they are used for interference cancellation in the next iteration.
The algorithm continues until the maximum number of iterations ($n_{\rm{iter}}$) is reached or the updated input symbol vector is close enough (less than $\epsilon$) to the previous one, as summarized in Algorithm \ref{algo:MRC_detection_lowcomp}.
Computing the complexity of Algorithm \ref{algo:MRC_detection_lowcomp} is straightforward as it only involves scalar operations. Step 3 to step 8 requires $2L$ CMs, $3L+1$ CAs and 1 CD. Therefore, its total complexity is $n_{\rm{iter}}(5L+1)(N-Q)$. 
\begin{algorithm}
  \SetAlgoLined
  \KwData{$\underline{\mathbf H}_{\mathrm{eff}}$, $d$, $\mathbf{y}$, $\hat{\mathbf{x}}^{(0)} = \mathbf{0}$, $\Delta \mathbf{y}^{(0)} = \mathbf{y}$}
   {\For{n = 1 : $n_{\rm{iter}}$}{
     \For{k = 0 : N-Q-1}{    
        
         $g_k^{(n)} =  \sum_{i = 0}^{L-1}\underline{ H}_{\mathrm{eff}}^{*}[{q}_k^i, k]\Delta y[{q}_k^i]^{(n)} + dx[k]^{(n - 1)}$

         $c_k^{(n)} = \frac{g_k^{(n)}}{d + \gamma^{-1}}$
         
         $\hat{x}[k]^{(n)} = c_k^{(n)}$
         
         \For{i = 0 : L-1}{
         $  \Delta y[{{q}_k^i}]^{(n)} = \Delta y[{{q}_k^i}]^{(n)} - \underline{ H}_{\mathrm{eff}}^{*}[{q}_k^i, k](x[k]^{(n)} - x[k]^{(n-1)})$
         }
        }
        \lIf{$||\hat{\mathbf{x}}^{(n)} - \hat{\mathbf{x}}^{(n-1)}|| < \epsilon$}{EXIT}
  }}
  \caption{Weighted MRC-based DFE detection}
  \label{algo:MRC_detection_lowcomp}
  
\end{algorithm}
\subsubsection{Convergence}
The detector convergence is analyzed using properties of iterative methods for linear systems. To this end, Algorithm \ref{algo:MRC_detection_lowcomp} can be expressed in the matrix form as
\begin{equation}
    \hat{\mathbf{x}}^{(n)} = \frac{d}{\gamma^{-1} + d}\hat{\mathbf{x}}^{(n-1)} + \frac{1}{\gamma^{-1} + d}(\underline{\mathbf H}_{\mathrm{eff}}^{H}\mathbf{y} - \mathbf{L}\hat{\mathbf{x}}^{(n)} - (\mathbf{L}^{H}+ \mathbf{D})\hat{\mathbf{x}}^{(n-1)}),\label{eq:x_nMatrixForm}
\end{equation}
where $\mathbf{D} = d\mathbf{I}$, $\mathbf{L}$ and $\mathbf{L}^H$ are the matrices containing diagonal elements, strictly lower and upper triangular parts  of the Hermitian matrix $\underline{\mathbf H}_{\mathrm{eff}}^{H}\underline{\mathbf H}_{\mathrm{eff}}$, respectively. Equation \eqref{eq:x_nMatrixForm} can be rewritten in the form
\begin{equation}
    \hat{\mathbf{x}}^{(n)} =-\mathbf{S}^{-1}(\mathbf{R} - \mathbf{S})\hat{\mathbf{x}}^{(n-1)} + \mathbf{M}^{-1}\mathbf{b},\label{eq:x_n_GS}
\end{equation}
where $\mathbf{S} = (\gamma^{-1} + d)\mathbf{I} + \mathbf{L}$, $\mathbf{R} =\underline{\mathbf H}_{\mathrm{eff}}^{H}\underline{\mathbf H}_{\mathrm{eff}}+\gamma^{-1}\mathbf{I}$ and $\mathbf{b} = \underline{\mathbf H}_{\mathrm{eff}}^{H}\mathbf{y}$.
The iteration in \eqref{eq:x_n_GS} is convergent if the spectral radius of the matrix $-\mathbf{S}^{-1}(\mathbf{R} - \mathbf{S})$, denoted as $\rho(-\mathbf{S}^{-1}(\mathbf{R} - \mathbf{S}))$, is strictly smaller than one \cite{bjorck1996numerical,saad2003iterative}.
\vspace{-3mm}
\begin{thm}\label{theo:convergence}
The iteration in \eqref{eq:x_n_GS} is convergent (i.e., $\rho(-\mathbf{S}^{-1}(\mathbf{R} - \mathbf{S})) < 1$) if $\mathbf{R} = \underline{\mathbf H}_{\mathrm{eff}}^{H}\underline{\mathbf H}_{\mathrm{eff}} + \gamma^{-1}\mathbf{I}$ is a positive definite Hermitian matrix.
\end{thm}
\begin{proof}
See Appendix \ref{sec:appen_convproof}.
\end{proof}
\subsubsection{Relation to the Gauss-Seidel method}\label{subsec:rel_GS}

LMMSE equalization is equivalent to solving the system of linear equations $(\underline{\mathbf H}_{\mathrm{eff}}^{H}\underline{\mathbf H}_{\mathrm{eff}}+\gamma^{-1}\mathbf{I})\mathbf{x} =\underline{\mathbf H}_{\mathrm{eff}}^{H}\mathbf{y}$. One way to solve it is using the properties of the Gauss-Seidel iterative method for solving linear equations \cite{bjorck1996numerical}. According to this method, decomposing $\underline{\mathbf H}_{\mathrm{eff}}^{H}\underline{\mathbf H}_{\mathrm{eff}}+\gamma^{-1}\mathbf{I}$ additively in its diagonal part $(d+\gamma^{-1}\mathbf{I})$, its strict lower triangular part $\mathbf{L}$ as well as its strict upper triangular part $\mathbf{L}^H$, gives $\mathbf{x}^{(n)}$ as
\begin{equation}
    \hat{\mathbf{x}}^{(n)} =-\mathbf{((\gamma^{-1} + d)\mathbf{I} + \mathbf{L})}^{-1}\mathbf{L}^H\mathbf{x}^{(n-1)} + \mathbf{((\gamma^{-1} + d)\mathbf{I} + \mathbf{L})}^{-1}\underline{\mathbf H}_{\mathrm{eff}}^{H}\mathbf{y}.\label{eq:GS_equivalent} 
\end{equation}
This means that the weighted MRC-based DFE converges to the LMMSE estimate. This is confirmed in the simulation results section. 
\vspace{-4mm}
\section{Embedded Channel Estimation}\label{section_channelEst}
In order to perform detection, the channel matrix $\mathbf{H}_{\rm{eff}}$ should be known at the receiver side. To enable that, we propose a channel estimation scheme based on the transmission of an embedded pilot symbol $x_{\rm pilot}$ in each AFDM frame surrounded by $Q$ null guard samples on each side of $x_{\rm{pilot}}$ (where $Q$ is defined in \eqref{eq:Q_val}) and $N-1-2Q$ data symbols $x_0^{\rm{data}},…,x_{N-2-Q}^{\rm{data}}$. The guard samples separate the data symbols from the pilot symbol so that the channel estimation can be done at the receiver without any interference from the data symbols. Equivalently, data detection using the estimated channel is performed without interference from the pilot symbol.

We place $x_{\rm{pilot}}$ as the first symbol in the frame as shown in Fig.~\ref{fig:FrameStrucure}:
\begin{figure}
\centering
\begin{minipage}{0.4\textwidth}
  \centering
  \includegraphics[width= \textwidth]{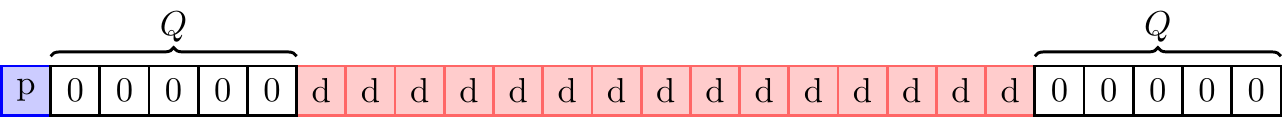}
  \caption{Symbol arrangement at the transmitter}
  \label{fig:FrameStrucure}
\end{minipage}
\hspace{0.01\textwidth}
\begin{minipage}{0.55\textwidth}
\centering
 \includegraphics[width= \textwidth]{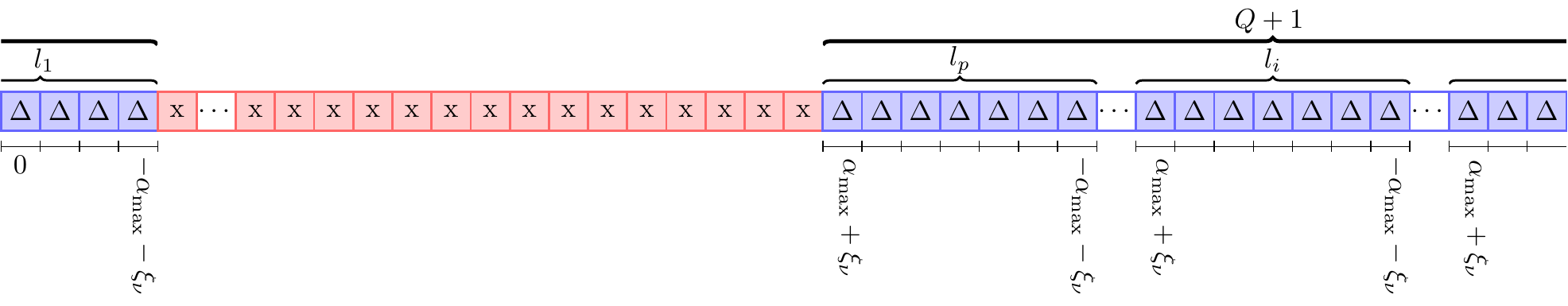}
    \caption{Received frame at the receiver}
    \label{fig:FrameStrucure_rec}
\end{minipage}
\end{figure}
\begin{equation}
    x[p] = \begin{cases}
    x_{\rm{pilot}}, \quad &p = 0\\
    0, \quad &1 \leq p \leq Q, N_Q+1 \leq p \leq N -1\\
    x_{p-Q-1}^{\rm{data}}, \quad &Q+1 \leq p \leq N_Q,
    \end{cases}
    \label{eq:x_with_pilot_sym}
\end{equation}
where $N_Q \triangleq N - Q - 1$. Considering the channel model defined in \eqref{eq:channel_model}, three parameters of each path, delay, Doppler shift, and complex gain, i.e., $3P$ unknown parameters $\boldsymbol{\theta} = [h_0, ..., h_{P-1}, l_0, ..., l_{P-1}, \nu_0, ..., \nu_{P-1}]$ should be estimated. As shown in Fig.~\ref{fig:FrameStrucure_rec}, the part of the received signal that is related to the pilot symbol is considered for the channel estimation. These symbols are expressed by the following equation
\vspace{-2mm}
\begin{equation}
    \underline{{\mathbf y}}_E = \underline{\mathbf H}_{\mathrm{eff}, E} \underline{\mathbf x}_E+ \widetilde{\underline{\mathbf w}}_E,
    \vspace{-1mm}
\end{equation}
where $E$ stands for "Estimation" phase, $\underline{\mathbf x}_E$, $\underline{\mathbf y}_E$ and $\underline{\mathbf H}_{\mathrm{eff}, E}$ are the parts of ${\mathbf x}$, ${\mathbf y}$ and $\mathbf{H}_{\mathrm{eff}}$ related to the channel estimation, respectively. They can be expressed by the matrix $\mathbf{T}_{t,E} = [\mathbf{I}_N]_{{\rm{ind}}_{t,E}
, :}$ and $\mathbf{T}_{r,E} = [\mathbf{I}_N]_{{\rm{ind}}_{r,E}, :}$ where ${\rm{ind}}_{t,E} = [0:Q \quad N_Q+1:N-1]$ and ${\rm{ind}}_{r,E} = [0:\alpha_{\max}+\xi_\nu \quad N_Q+\alpha_{\max}+\xi_\nu+1:N-1]$ as $\underline{\mathbf x}_E = \mathbf{T}_{t, E}\mathbf{x}$, $\underline{\mathbf y}_E = \mathbf{T}_{r, E}\mathbf{y}$ and  $\underline{\mathbf H}_{\mathrm{eff}, E} = \mathbf{T}_{r, E}{\mathbf H}_{\mathrm{eff}} \mathbf{T}_{t, E}^H$. Considering the ML detector, the log-likelihood function to be minimized is given by 
\begin{equation}
    l(\underline{{\mathbf y}}_E|\boldsymbol{\theta},\underline{\mathbf x}_E ) = \|\underline{{\mathbf y}}_E -  \underline{\mathbf H}_{\mathrm{eff}, E} \underline{\mathbf x}_E\|^2.
\end{equation}
Considering \eqref{eq:diff_paths}, it can be written as 
\vspace{-2mm}
\begin{equation}
    l(\underline{{\mathbf y}}_E|\boldsymbol{\theta},\underline{\mathbf x}_E ) = \|\underline{{\mathbf y}}_E -  \sum_{i = 0}^{P-1}h_i\underline{\mathbf H}_{i, E} \underline{\mathbf x}_E\|^2.\label{eq:log_liklihood_func}
    \vspace{-1mm}
\end{equation}
where $\underline{\mathbf H}_{i, E} = \mathbf{T}_{r, E}{\mathbf H}_{i} \mathbf{T}_{t, E}^H$. Since $\underline{\mathbf x}_E$ has only one non-zero element at the first entry, $\underline{\mathbf H}_{i, E} \underline{\mathbf x}_E = x_{\rm{pilot}}\underline{\mathbf{h}}_{i, E, 1}$ where $\underline{\mathbf{h}}_{i, E, 1}$ is the first column of $\underline{\mathbf H}_{i, E}$ and is thus dependent on $l_i$ and $\nu_i$ as can be seen from \eqref{eq:Hi_p_q_integer} and \eqref{eq:Hi_p_q_fractional}, i.e., $\underline{\mathbf{h}}_{i, E, 1}= \underline{\mathbf{h}}_{i, E, 1}(l_i,\nu_i)$ or equivalently $\underline{\mathbf{h}}_{i, E, 1}= \underline{\mathbf{h}}_{i, E, 1}(l_i,\alpha_i,a_i)$. Thus, the ML estimator is given by 
\begin{equation} 
\hat{\boldsymbol{\theta}} = {\rm{arg}} \min _{\boldsymbol{\theta} \in \mathbb{C}^P\times\mathbb{R}^P\times \mathbb{R}^P}\|\underline{{\mathbf y}}_E -  x_{\rm{pilot}}\sum_{i = 0}^{P-1} h_i\underline{\mathbf{h}}_{i, E, 1}(l_i,\nu_i) \|^2.\label{eq:ChEst_ML}
\end{equation}
As brute force search is infeasible in a $3P$-dimensional continuous domain, we propose a low complexity solution to \eqref{eq:ChEst_ML}. For given $\{l_i, \nu_i\}$, the log-likelihood function in \eqref{eq:log_liklihood_func} is quadratic in the complex gain $h_i$. Therefore, solving \eqref{eq:ChEst_ML} with respect to $h_i$ leads to the linear system of equations
\begin{equation}
    \sum_{j = 0}^{P-1}h_j\underline{\mathbf{h}}^H_{i, E, 1}(l_i,\nu_i)\underline{\mathbf{h}}_{j, E, 1}(l_j,\nu_j) = \frac{\underline{\mathbf{h}}_{i, E, 1}^H(l_i,\nu_i)\underline{{\mathbf y}}_E}{x_{\rm{pilot}}}, \quad i = 0, 1, \cdots, P-1.\label{eq:ML_eqs}
\end{equation}
Expanding \eqref{eq:ChEst_ML} and using \eqref{eq:ML_eqs}, the minimization with respect to the $\{l_i, \nu_i\}$ reduces to maximizing the function
\begin{equation}
    l_2(\underline{{\mathbf y}}_E|\boldsymbol{\theta},x_{\rm{pilot}} ) = \sum_{i=0}^{P-1}\frac{|\underline{\mathbf{h}}_{i, E, 1}^H(l_i,\nu_i)\underline{{\mathbf y}}_E|^2}{\underline{\mathbf{h}}_{i, E, 1}^H(l_i,\nu_i)\underline{\mathbf{h}}_{i, E, 1}(l_i,\nu_i)} - \frac{(\sum_{j \neq i }h_j\underline{\mathbf{h}}^H_{i, E, 1}(l_i,\nu_i)\underline{\mathbf{h}}_{j, E, 1}(l_j,\nu_j))\underline{{\mathbf y}}_E^H\underline{\mathbf{h}}_{i, E, 1}x_{\rm{pilot}}}{\underline{\mathbf{h}}_{i, E, 1}^H(l_i,\nu_i)\underline{\mathbf{h}}_{i, E, 1}(l_i,\nu_i)}.\label{eq:delaydop_ML}
\end{equation}
Now considering \eqref{eq:ML_eqs} and \eqref{eq:delaydop_ML}, we show how channel estimation is performed for the integer and fractional Doppler shift cases in the following subsections.
\vspace{-1mm}
\subsection{Integer Doppler Case}
In this case ($\nu_i = \alpha_i$), as it can be seen from \eqref{eq:Hi_p_q_integer}, ${\underline{\mathbf{h}}_{i, E, 1}}$ has only one non-zero element. In addition, for different paths, the location of these non-zero elements are different from each other, i.e,
\vspace{-2mm}
\begin{equation}
    \underline{\mathbf{h}}_{i, E, 1}^H(l_i,\alpha_i)\underline{\mathbf{h}}_{j, E, 1}(l_j,\alpha_j) = \begin{cases}
    1, & i = j\\
    0, & i\neq j
    \end{cases}.\label{eq:path_orthogonality}
\end{equation}
Thus, \eqref{eq:ML_eqs} and \eqref{eq:delaydop_ML} are rewritten as
\begin{align}
        &h_i = \frac{\underline{\mathbf{h}}_{i, E, 1}^H(l_i,\alpha_i)\underline{{\mathbf y}}_E}{x_{\rm{pilot}}}, \quad i = 0, 1, \cdots, P-1,\label{eq:ML_eqs_integer}\\
        & l_2(\underline{{\mathbf y}}_E|\boldsymbol{\theta},x_{\rm{pilot}} ) = \sum_{i = 0}^{P-1}|\underline{\mathbf{h}}_{i, E, 1}^H(l_i,\alpha_i)\underline{{\mathbf y}}_E|^2\label{eq:delaydopML_integer},
        \vspace{-5mm}
\end{align}
respectively. Thus, the delays and Doppler shifts, i.e., $\mathbf{l} \triangleq [l_0, ..., l_{P-1}]$ and $\boldsymbol{\alpha} \triangleq [\alpha_0, ..., \alpha_{P-1}]$ can be estimated as the argument maximizing the r.h.s. of \eqref{eq:delaydopML_integer}. Due to the structure of $\underline{\mathbf{h}}_{i, E, 1}(l_i,\alpha_i)$, maximizing the r.h.s. of  \eqref{eq:delaydopML_integer} is equivalent to finding the indices of the largest entries of $\underline{\mathbf{y}}_E$. After finding the pair of parameters $\{l_i, \alpha_i\}$ for all the paths, the paths complex gains can be obtained using  \eqref{eq:ML_eqs_integer}.

\vspace{-2mm}
\subsection{Fractional Doppler Case}
For the fractional case ($\nu_i = \alpha_i + a_i$), as it is shown in \eqref{eq:Hi_p_q_fractional}, \eqref{eq:path_orthogonality} cannot in theory hold. Therefore, it is impossible to directly maximize \eqref{eq:delaydop_ML} as the complex gains $h_i$ are not known. Moreover, the second term in \eqref{eq:delaydop_ML} depends on all pairs of $\{l_i, \nu_i\}$ for $j \neq i$. However, assuming large enough $\xi_\nu$, the value of $\underline{\mathbf{h}}_{i, E, 1}^H(l_i,\alpha_i,a_i)\underline{\mathbf{h}}_{j, E, 1}(l_j,\alpha_j,a_j)$ is very small when $i\neq j$. Thus, for the fractional case, we exploit this approximation and assume that \eqref{eq:path_orthogonality} holds to maximize \eqref{eq:delaydop_ML}. 
With this assumption, in order to find the $\{l_i, \nu_i\}$ pairs that maximize \eqref{eq:delaydop_ML}, first, we find the delay and integer part of the Doppler shift, i.e, $\mathbf{l}$ and $\boldsymbol{\alpha}$. To this end, we denote all the delays and integer Doppler shifts combinations set by {$\mathcal{L}\triangleq\{(\mathbf{l},\boldsymbol{\alpha}) | 0\leq l[i]\leq l_{\max}, -\alpha_{\max}\leq \alpha[i]\leq\alpha_{\max}\}$} and pick the one that maximizes \eqref{eq:delaydop_ML}. In the next phase, the fractional parts $\mathbf{a} \triangleq [a_0, ..., a_{P-1}]$ are estimated using the obtained delay and integer Doppler shifts as
\begin{equation}
\hat{\mathbf{a}} = \argmax_{\mathbf{a}} \sum_{i=0}^{P-1}\frac{|\underline{\mathbf{h}}_{i, E, 1}^H(\hat{l}_i,\hat{\alpha}_i,a_i)\underline{{\mathbf y}}_E|^2}{\underline{\mathbf{h}}_{i, E, 1}^H(\hat{l}_i,\hat{\alpha}_i,a_i)\underline{\mathbf{h}}_{i, E, 1}(\hat{l}_i,\hat{\alpha}_i,a_i)},    
\end{equation}
using a search on fine discretization of $[-1/2,1/2]^P$. Then, the estimated Doppler shifts become $\hat{\boldsymbol{\nu}} \triangleq [\hat{\nu}_0, ..., \hat{\nu}_{P-1}]= \hat{\boldsymbol{\alpha}} + \hat{\mathbf{a}}$. The complex gains are estimated by solving the linear system \eqref{eq:ML_eqs}. It is worth noting that this algorithm has good performance if the paths have different delays, i.e., for each delay, there is only one Doppler shift.
\section{Simulation Results}\label{sec_sim_result}
In this section, we provide simulation results to assess the performance of AFDM. In all simulations, the complex gains $h_i$ are generated as independent complex Gaussian random variables with zero mean and $1/P$ variance. The carrier frequency is $4$ GHz. BER values are obtained using $10^6$ different channel realizations.

Fig.~\ref{fig:AFDM_diversity} shows the simulated BER performance of AFDM for different channels with $N = 16$ and BPSK using the ML detection. 
We consider three different channels with different numbers of paths, namely a 2-path, a 3-path , and a 4-path channel. The maximum delay spread (in terms of integer taps) is set to be $2$ ($l_{\max} = 1$), $3$ ($l_{\max} = 2$) and $4$ ($l_{\max} = 3$), respectively. The duration between two successive delay taps is approximately $41.6$ $\mu$s. The maximum Doppler shift is considered $\alpha_{\max} = 1$, which corresponds to a maximum speed of $405$ km/h. We observe that for each channel, AFDM achieves the full diversity of the channel. Note that the plots of $s_1(SNR)^{−2}$, $s_2(SNR)^{−3}$ and $s_3(SNR)^{−4}$ are only used to identify the slope of the curves and do not represent an upper bound. 

\begin{figure}
\centering
\begin{subfigure}{0.4\textwidth}
  \centering
  \includegraphics[width= \textwidth]{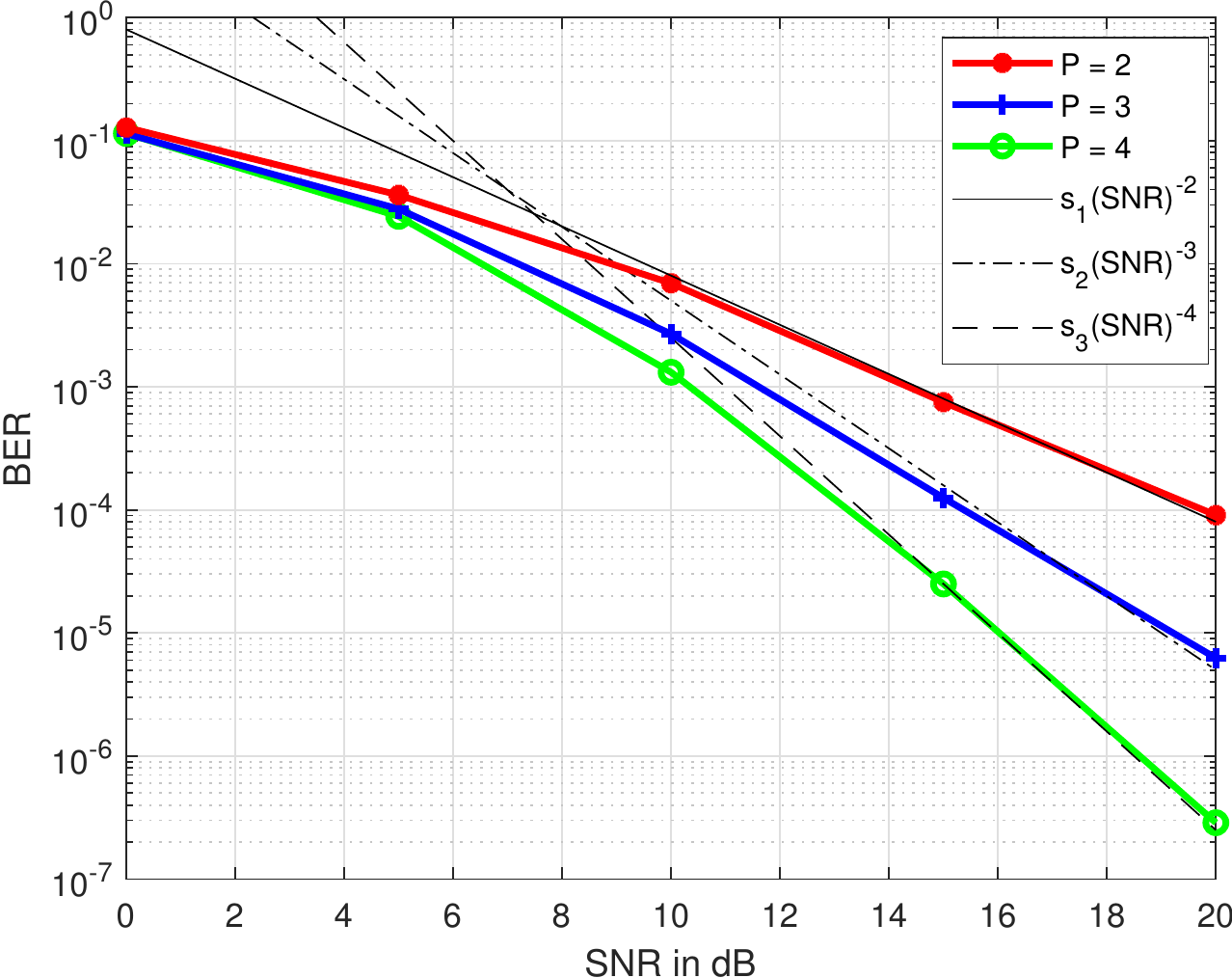}
    \caption{BER performance of AFDM with different number of paths for $N = 16$ .}
    \label{fig:AFDM_diversity}
\end{subfigure}
\hspace{0.1\textwidth}
\begin{subfigure}{0.4\textwidth}
\centering
 \includegraphics[width= \textwidth]{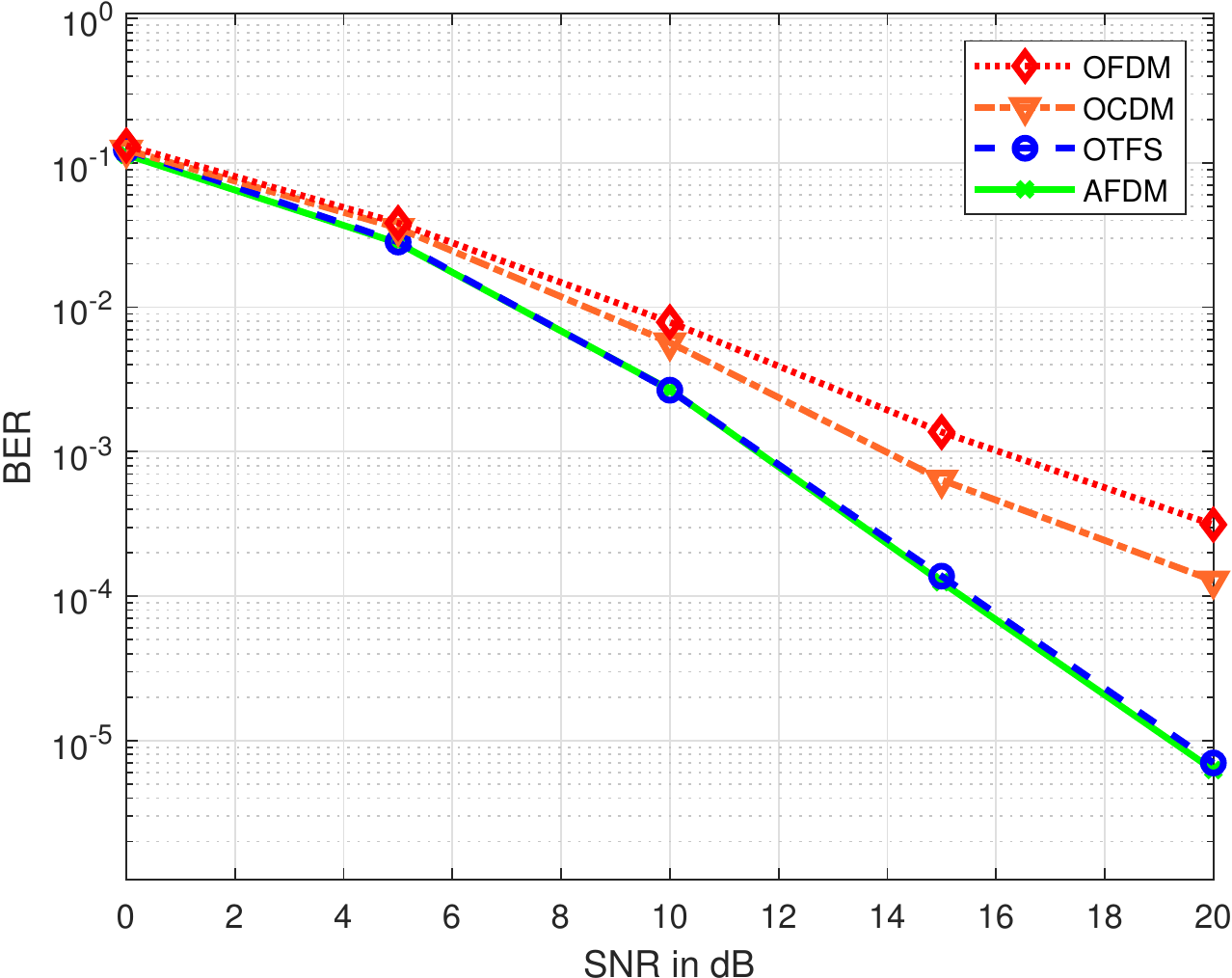}
    \caption{BER performance of OFDM, OCDM, OTFS and AFDM in a three-path channel for $N = 16$,  $N_{\rm{OTFS}} = 4$ and $M_{\rm{OTFS}} =4$.}
    \label{fig:BERDiffSchmemes}
\end{subfigure}
\caption{BER performance using BPSK in LTV channels using ML detection.}
\label{}
\end{figure}

Before proceeding further, we recall that in OTFS the time–frequency signal plane is sampled at intervals $T_{\rm{OTFS}}$ (seconds) and $\Delta f_{\rm{OTFS}}$ (Hz), respectively to obtain a grid as 
\begin{equation}
\Lambda_{\rm{OTFS}} = \{(nT_{\rm{OTFS}},m\Delta f_{\rm{OTFS}}), n = 0,\cdots,N_{\rm{OTFS}} − 1,m = 0,\cdots, M_{\rm{OTFS}} − 1\},    
\end{equation}
and modulated time–frequency samples $X[n,m], n = 0, . . .,N_{\rm{OTFS}} − 1, m = 0, . . .,M_{\rm{OTFS}} − 1$ are transmitted over an OTFS frame with duration $T_{f_{\rm{-OTFS}}} = N_{\rm{OTFS}}T_{\rm{OTFS}}$ and occupy a bandwidth $B_{\rm{OTFS}} = M_{\rm{OTFS}}\Delta f_{\rm{OTFS}}.$
In order to compare the performance of AFDM against OTFS, we assume $N = M_{\rm{OTFS}}N_{\rm{OTFS}}$ to have the same resources in AFDM and OTFS frames.

Fig.~\ref{fig:BERDiffSchmemes} shows the BER performance of AFDM, OFDM, OCDM, and OTFS. For the DAFT-based schemes, we generate the frames with $N = 16$. 
OTFS frame is generated with $N_{\rm{OTFS}} = 4$ and $M_{\rm{OTFS}} = 4$.
The maximum delay spread is set to be $l_{\max} = 2$ and the maximum Doppler shift is $\alpha_{\max} = 1$. The delay shifts are fixed and Jakes Doppler spectrum is considered for each channel realization, i.e, the Doppler shifts are varying and the Doppler shift of each path is generated using $\alpha_i = \alpha_{\max} cos(\theta_i)$, where $\theta_i$ is uniformly distributed over $[−\pi, \pi]$. Expectedly, OFDM has the worse performance as it cannot separate the paths. The performance of OCDM depends on the delay-Doppler profile of the channel. OCDM performs poorly and has the same diversity (one) as OFDM, due to the possible destructive addition of the two overlapping paths.
The reason why OCDM has better performance than OFDM is related to its better path separation capabilities than OFDM. The proposed AFDM achieves full diversity, mainly due to path separation by tuning $c_1$ and setting $c_2$ to be an arbitrary irrational number or a rational number sufficiently smaller than $\frac{1}{2N}$. We also observe that AFDM has the same BER performance as OTFS.

In the previous figures, small $N$ values are assumed along with ML detection to show the diversity order of AFDM. In the following figures, we consider a more practical configuration with QPSK, $N=256$ for the DAFT-based schemes and $N_{\rm{OTFS}} = 16$, $M_{\rm{OTFS}} = 16$ for the OTFS. We consider a 3-path channel. The maximum Doppler shift is $\alpha_{\max} = 2$, which corresponds to a speed of $540$ km/h, and the Doppler shift of each path is generated using Jakes Doppler spectrum.  The maximum delay spread is set to be $l_{\max} = 2$.
Fig.~\ref{fig:MMSEdetection} shows the BER performance of the DAFT-based schemes and OTFS. All results are obtained with LMMSE detection at the receiver. We observe that AFDM outperforms OFDM and OCDM, while having identical performance with OTFS. However, when channel estimation is taken into account, the pilot overhead of OTFS is twice that of AFDM due to the 2D structure of its underlying transform. Indeed, while the AFDM embedded pilot scheme presented in Section \ref{section_channelEst} occupies $2(l_{\max} + 1)(2(\alpha_{\max} + \xi_\nu) + 1)-1$ entries out of the $N$ entries of the AFDM symbol, its OTFS counterpart \cite{raviteja2019embedded} requires $\left(4(\alpha_{\max} + \xi_\nu)+1\right)\left(2l_{\max}+1\right)$ (for the integer Doppler shifts $\xi_\nu = 0 $ ). This difference translates into a significant advantage of AFDM over OTFS in terms of spectral efficiency, as shown in Fig.~\ref{fig:throughput}. The spectral efficiency values were derived from the BER values plotted in Fig.~\ref{fig:MMSEdetection}.

\begin{figure}
\centering
\begin{subfigure}{0.3\textwidth}
\centering
 \includegraphics[scale=.50]{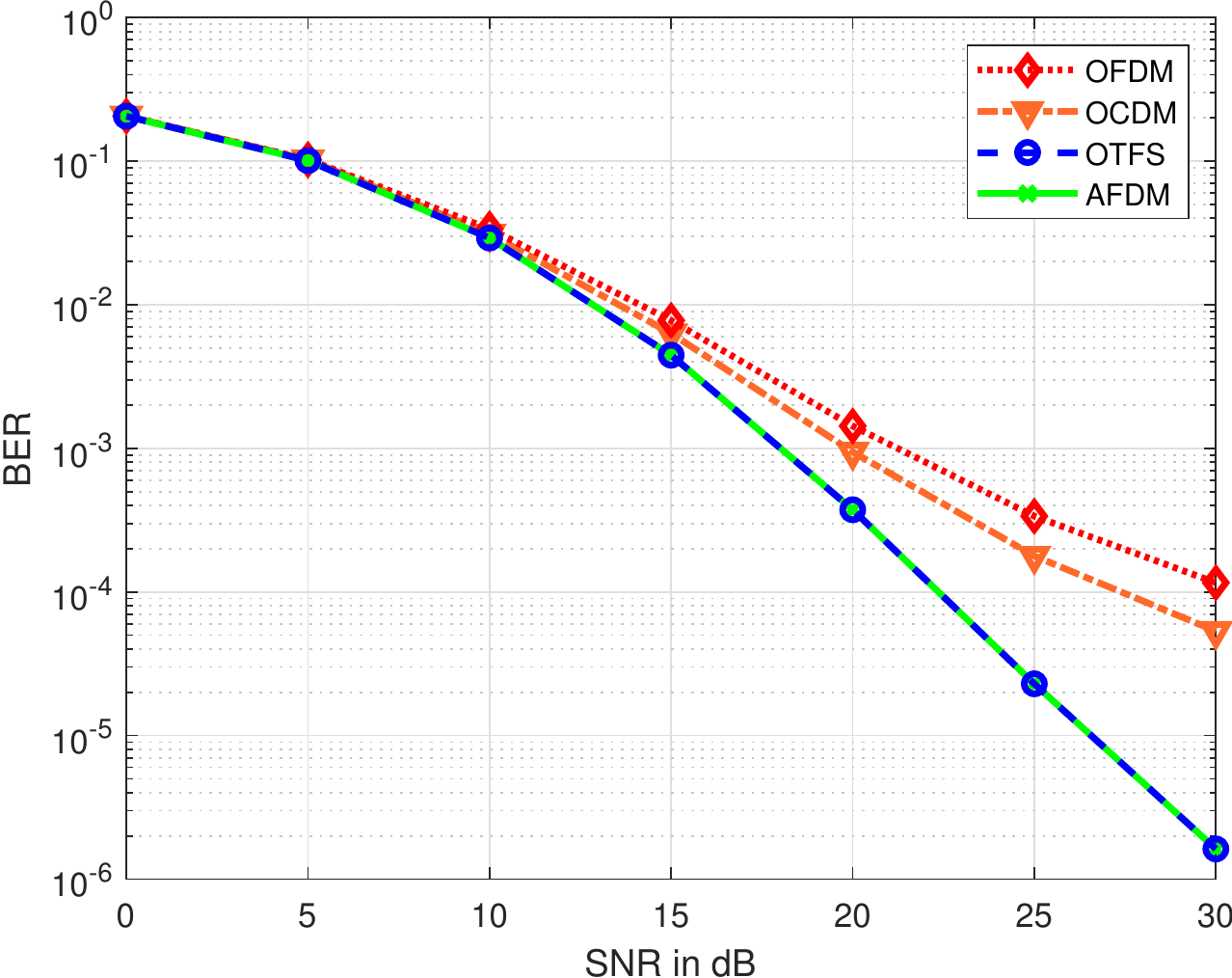}
    \caption{BER }
    \label{fig:MMSEdetection}
\end{subfigure} \hspace{0.2\textwidth}
\begin{subfigure}{0.3\textwidth}
\centering
\includegraphics[scale=.50]{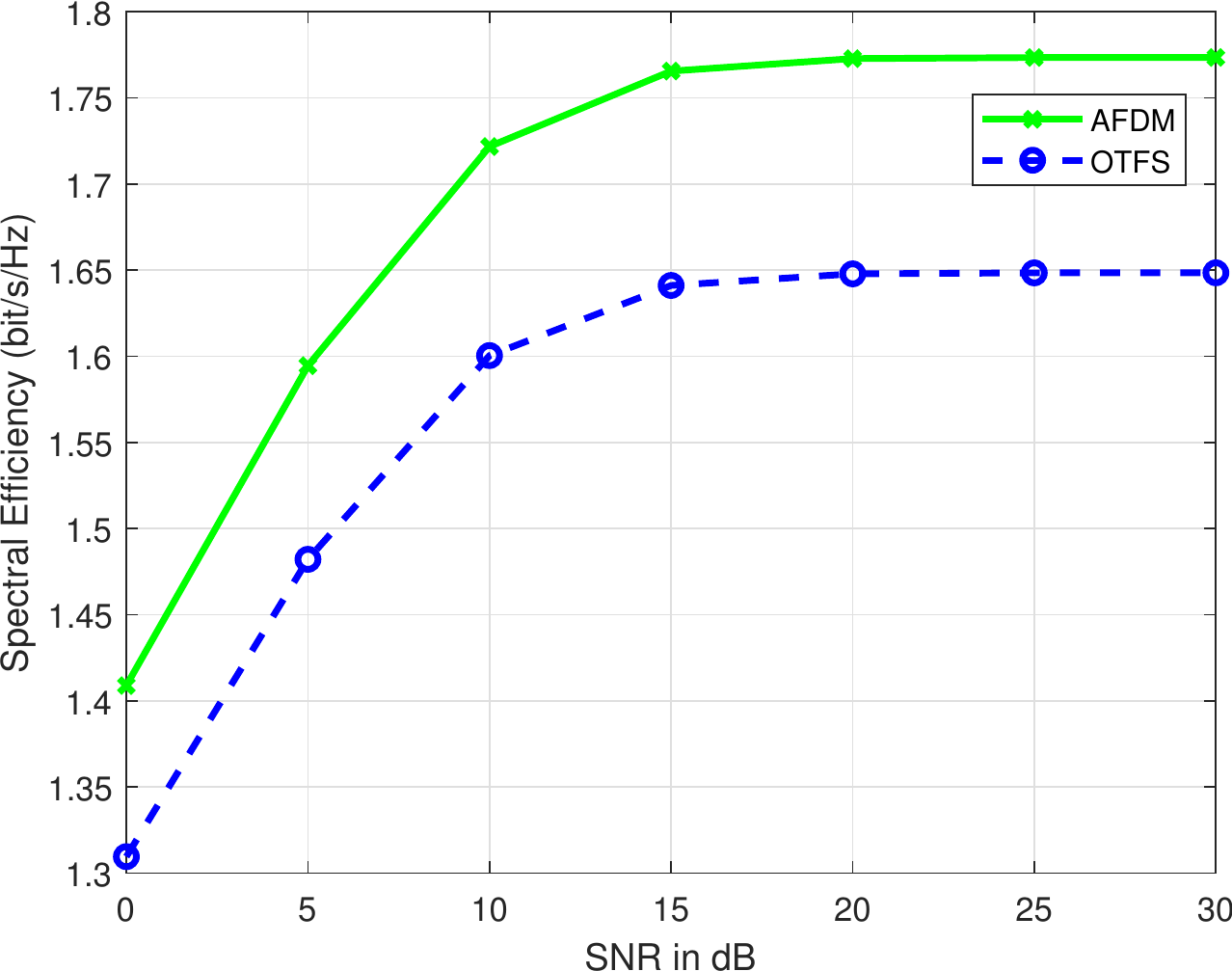}
    \caption{Spectral efficiency}
    \label{fig:throughput}
\end{subfigure}
\caption{BER and spectral efficiency performance of OFDM, OCDM, OTFS and AFDM using MMSE detection.}
\label{fig:AFDM_OTFS_BER_SE}
\vspace{-4mm}
\end{figure}

Fig.~\ref{Fig:detection} compares the performance of AFDM and OFDM in terms of BER using different detectors. In this figure, integer and fractional Doppler shifts are considered. We observe that AFDM outperforms OFDM, owed to achieving full diversity and every information symbol being received through multiple independent non-overlapping paths.
Moreover, it shows that the weighted MRC-based DFE has close performance to exact LMMSE, which validates Theorem \ref{theo:convergence}. In Fig.~\ref{Fig:detection_integer}, all BER curves (i.e., of the two methods and of low-complexity MMSE \cite{bemani2022low} based on banded matrix approximation) coincide because the channel matrices used in the three detection methods are all the same and are banded without approximation. In Fig.~\ref{Fig:detection_frac}, exact LMMSE has slightly better performance than the low-complexity methods due to the use of the banded-matrix approximation in the latter when Doppler frequency shifts are fractional.

\begin{figure}
\centering
\begin{subfigure}{0.3\textwidth}
\centering
 \includegraphics[scale=.50]{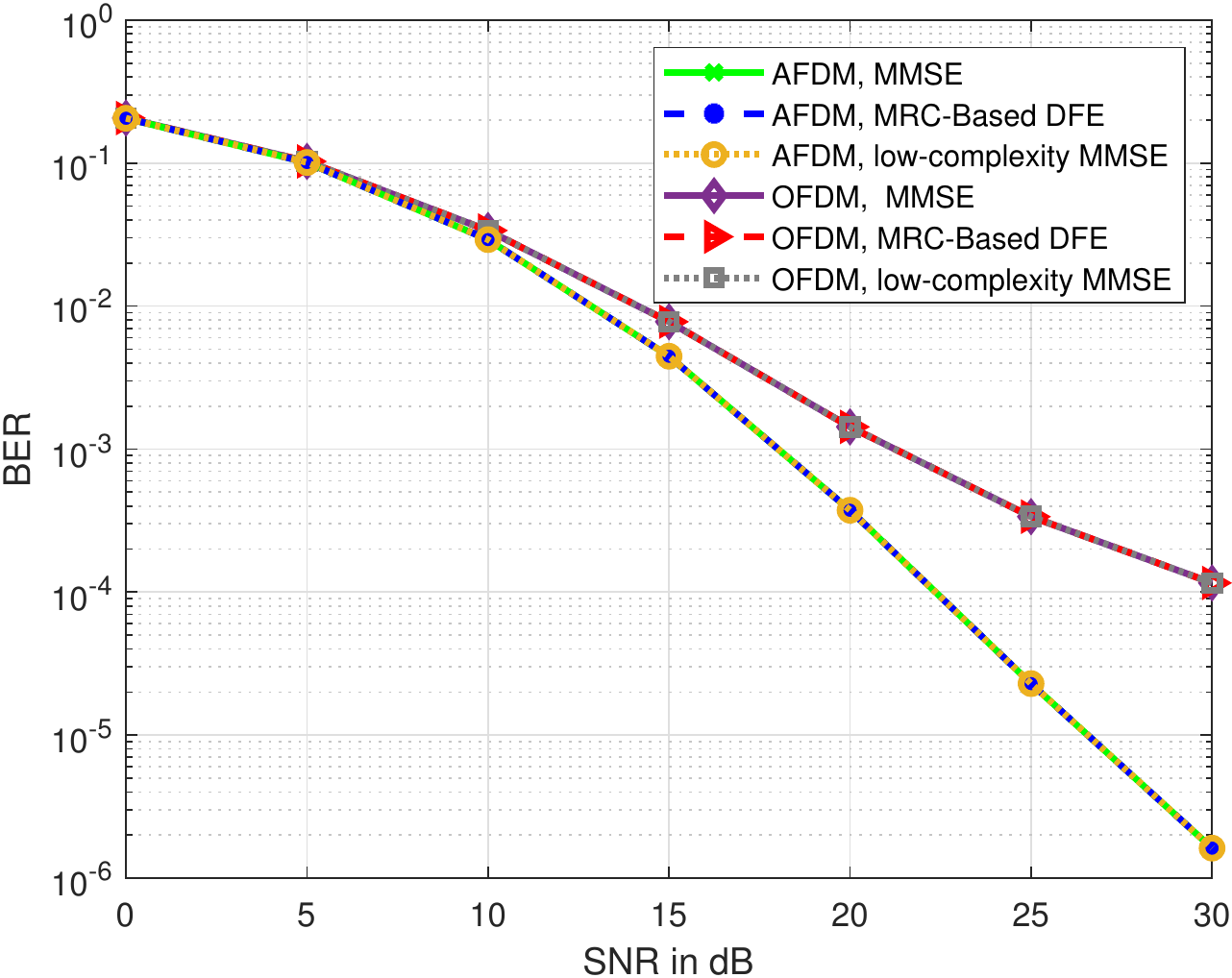}
    \caption{Integer Doppler shifts}
    \label{Fig:detection_integer}
\end{subfigure} \hspace{0.2\textwidth}
\begin{subfigure}{0.3\textwidth}
\centering
\includegraphics[scale=.50]{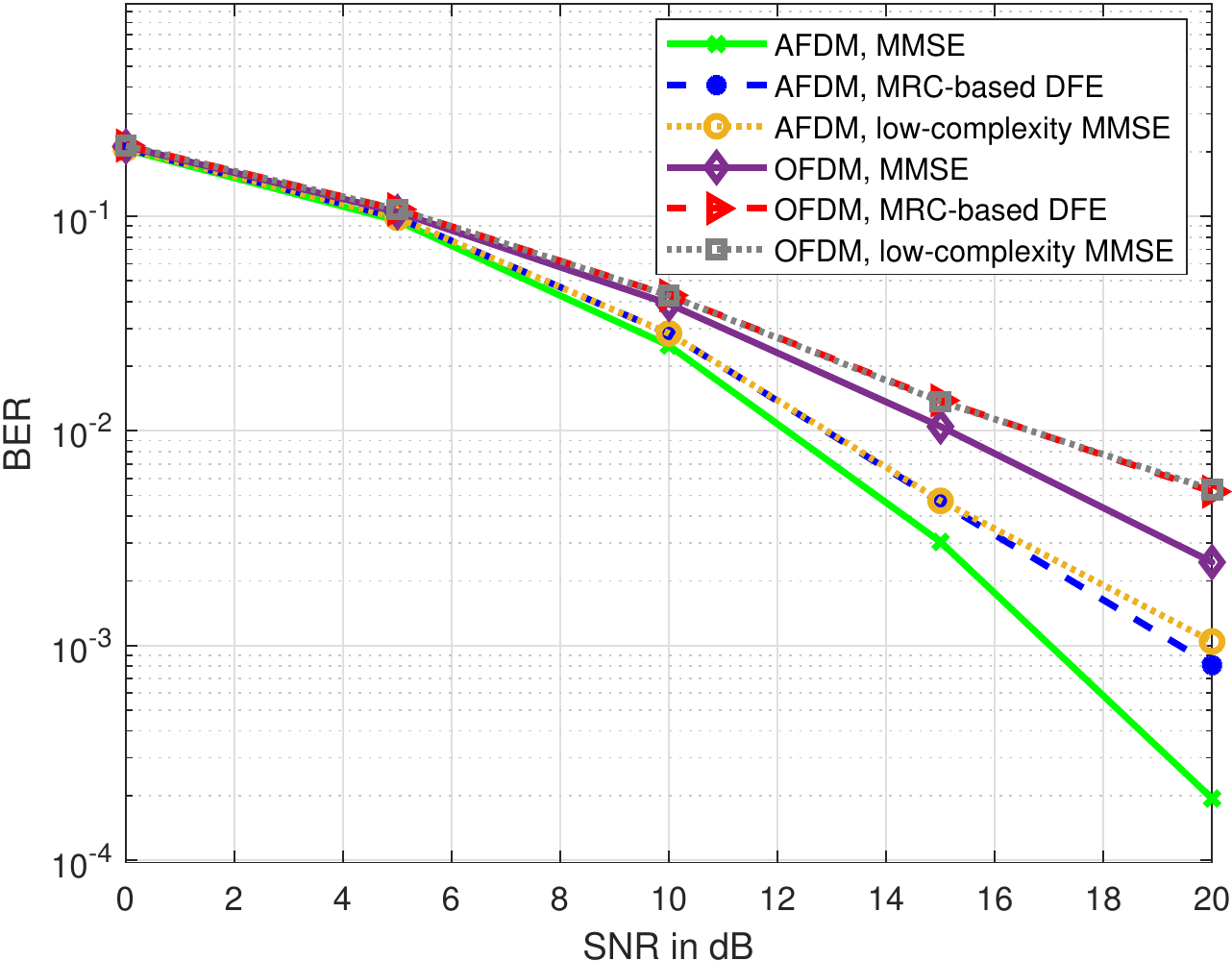}
    \caption{Fractional Doppler shifts}
    \label{Fig:detection_frac}
\end{subfigure}
\caption{BER performance comparison between AFDM and OFDM systems using different detectors for the integer and fractional Doppler shifts.} 
\label{Fig:detection}
\vspace{-5mm}
\end{figure}

We now assess the BER performance of AFDM when detection is performed based on the channel state information given by the proposed channel estimation scheme. The pilot symbol SNR is denoted by ${\rm{SNR}}_{\rm p} = \frac{|x_{\rm{pilot}}|^2}{N_0}$ and the data symbols have the ${\rm{SNR}}_{\rm d} = \frac{\mathbb{E}\{|x^{\rm{data}}|^2\}}{N_0}$. 
\begin{figure}
    \centering
    \includegraphics[scale=.55]{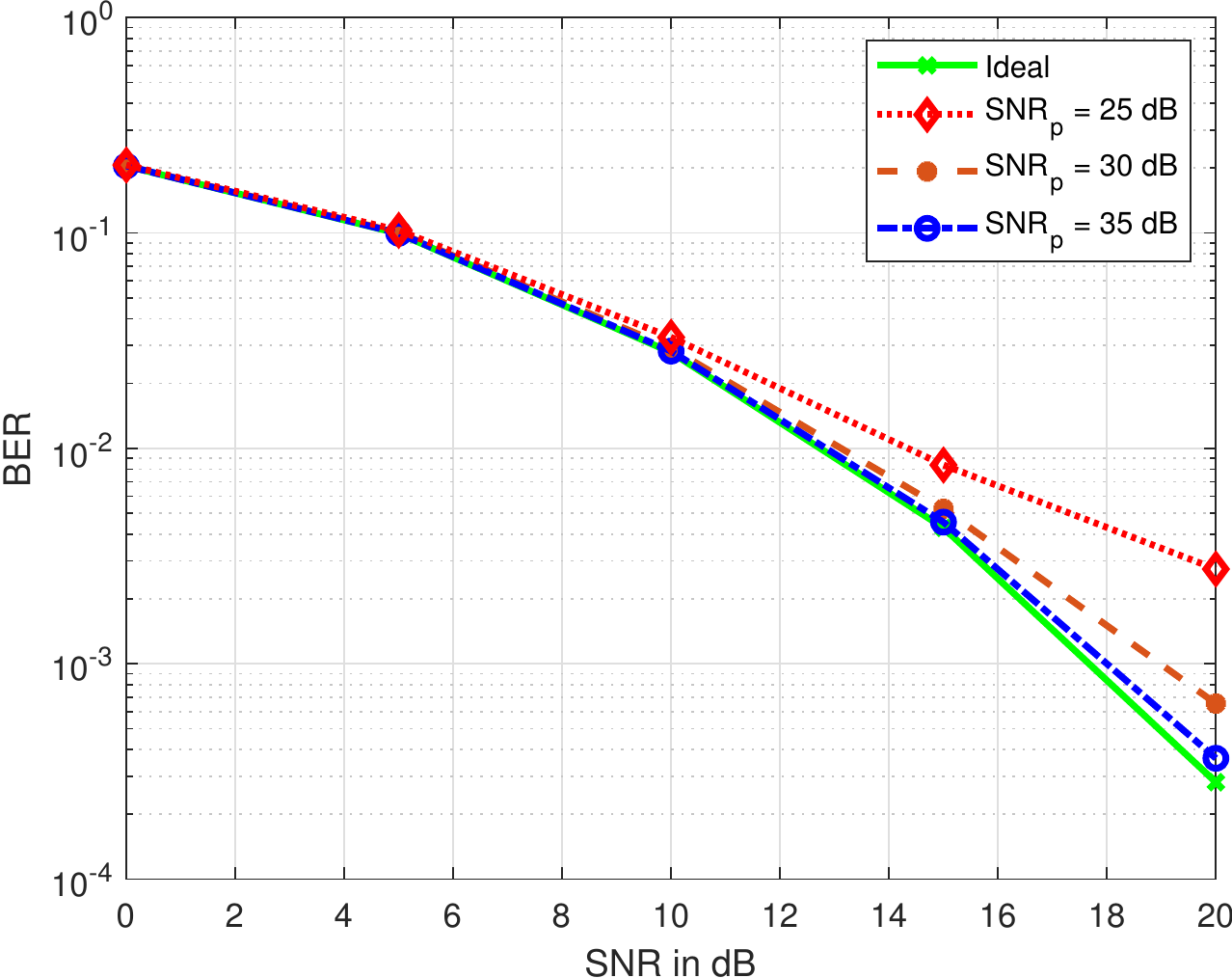}
    \caption{BER versus ${\rm{SNR}}_{\rm d}$ for the integer Doppler case with different ${\rm{SNR}}_{\rm p}$ and ideal channel.}
    \label{fig:channel_est_integer}
    \vspace{-5mm}
\end{figure}
Fig.~\ref{fig:channel_est_integer} shows the BER versus ${\rm{SNR}}_{\rm d}$ for AFDM considering the ideal case of perfect channel knowledge at the receiver as well as the case where the channel is estimated using the proposed algorithm for integer Doppler case with different values of ${\rm{SNR}}_{\rm p}$. As ${\rm{SNR}}_{\rm p}$ increases, the BER decreases and the AFDM performance improves. Moreover, we see that for ${\rm{SNR}}_{\rm p} = 35$ dB, the performance of AFDM with the propsoed channel estimation is very close to the ideal case.

Fig.~\ref{fig:ChEst_frac_snr} shows the BER performance of AFDM for different ${\rm{SNR}}_{\rm p}$ considering the fractional Doppler shift case. Similar to the integer Doppler shift case, increasing the pilot power improves the error performance. As we can see, with ${\rm{SNR}}_{\rm p} = 40$ dB, AFDM with the proposed embedded channel estimation has similar performance with AFDM with perfect channel knowledge at the receiver. Note that the system has more overhead in the fractional Doppler shift case. In addition, larger ${\rm SNR}_{\rm p}$ is needed to achieve the same performance. Note that in practice, it is possible to assume larger values for ${\rm SNR}_{\rm p}$ compared to ${\rm SNR}_{\rm d}$ since the zero guard samples surrounding the pilot symbol allow for the transmit power of the latter to be boosted without violating the average transmit power constraint.
Under ideal receiver-side channel knowledge, it can be seen from Fig.~\ref{fig:ChEst_frac_xi} that increasing $\xi_{\nu}$, improves the performance of AFDM, as less overlapping is occurring between the matrices ${\mathbf H}_{i}$ belonging to different channel paths in the effective channel matrix ${\mathbf H}_{\mathrm{eff}}$. With practical channel estimation, Fig.~\ref{fig:ChEst_frac_xi} shows that increasing $\xi_{\nu}$ also improves the channel estimation quality, since inter-path interference when channel estimation algorithm for fractional Doppler shift case is performed decreases for the same reason, i.e, less overlapping between the matrices ${\mathbf H}_{i}$.

\begin{figure}
\centering
\begin{subfigure}{0.3\textwidth}
\centering
 \includegraphics[scale=.50]{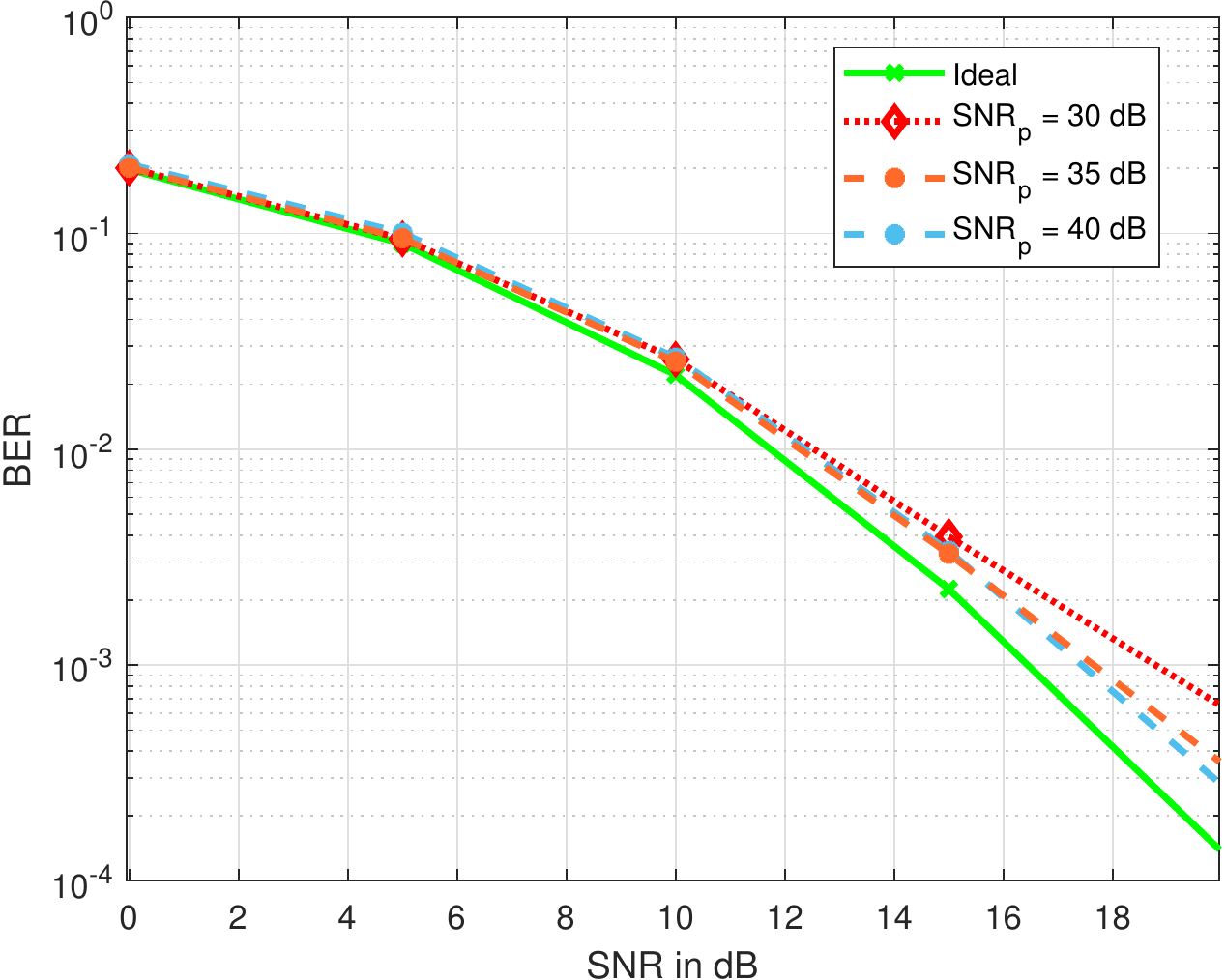}
    \caption{BER versus ${\rm{SNR}}_{\rm d}$ considering different ${\rm{SNR}}_{\rm p}$ and ideal receiver-side CSI}
    \label{fig:ChEst_frac_snr}
\end{subfigure} \hspace{0.2\textwidth}
\begin{subfigure}{0.3\textwidth}
\centering
\includegraphics[scale=.50]{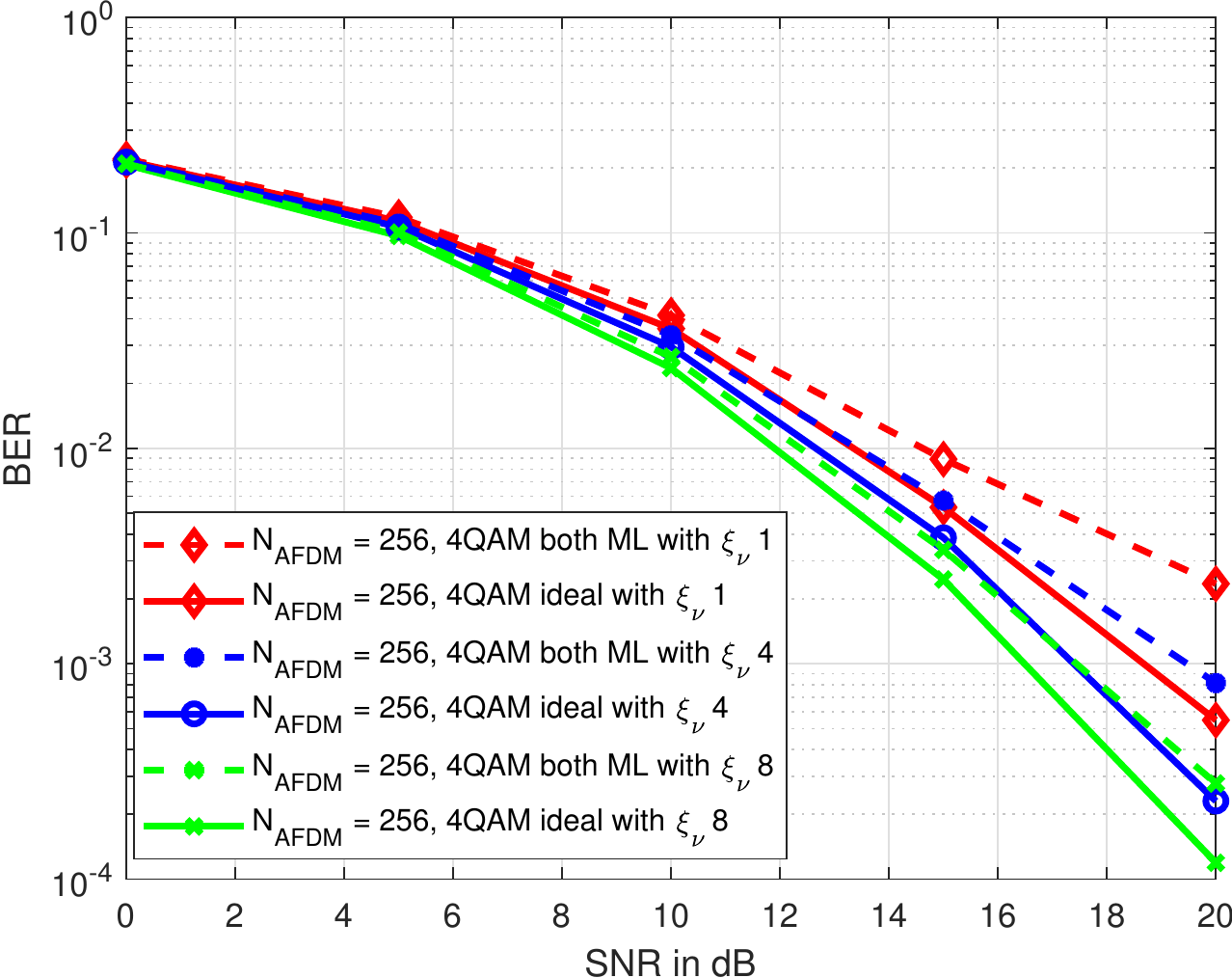}
    \caption{BER versus $SNR_d$ with different $\xi_\nu$ considering ideal and estimated channel with $ {\rm{SNR}}_{\rm p} = 40$ dB}
    \label{fig:ChEst_frac_xi}
\end{subfigure}
\caption{The effect of ${\rm{SNR}}_{\rm p}$ and $\xi_\nu$ on the BER performance}
\label{fig:ChEst}
\vspace{-2mm}
\end{figure}
    \vspace{-2mm}
\section{Conclusion}\label{sec_conc}
We proposed a new waveform, coined AFDM, which employs multiple discrete-time orthogonal chirp signals generated using the discrete affine Fourier transform. The unique features and effects of DAFT were revealed by deriving the input-output relation. Using the input-output relation, the AFDM parameters can be tuned such that the DAFT domain channel impulse response constitutes a full representation of its delay-Doppler profile. Then, we showed analytically that AFDM can achieve full diversity in doubly dispersive channel by properly tuning its pulse parameters. Inserting zero-padding in the DAFT domain, we proposed a low complexity detector and channel estimation algorithms for AFDM. Simulation results showed that AFDM outperforms OFDM and other DAFT-based multicarrier schemes, while having advantages over OTFS in terms of pilot and user multiplexing overhead. The main takeaway of this paper is that AFDM is a promising new waveform for high mobility communications in future wireless systems.

\appendices

\vspace{-4mm}
\section{Proof of theorem \ref{theo:main}}\label{sec_appen_proof}
We give the proof of Theorem \ref{theo:main} only in the case of integer Doppler shifts. The proof holds also for the fractional Doppler shift with some modifications.

First, we show that when \eqref{eq:opt_c1_int} and \eqref{eq:cond_1} hold, there exist values of $c_2$ such that the rank of $\mathbf{\Phi}(\boldsymbol{\delta})$ is equal to $P$, i.e., such that the $P$ columns of $\mathbf{\Phi}(\boldsymbol{\delta})$ are linearly independent. Therefore, considering the $\mathbf\Phi(\boldsymbol{\delta})$ for a $P$-path channel
{\small
\begin{equation}
 \begin{aligned}
    &{\boldsymbol{\Phi}}({\boldsymbol{\delta}}) = [\mathbf{H}_1\boldsymbol{\delta} \mid \ldots \mid \mathbf{H}_P\boldsymbol{\delta}]=\\
    &{\left[\begin{array}{cccc}{H}_{eff}[0, \mathrm{loc}_1]\mathbf{\delta}[{\mathrm{loc}_1}]& \cdots & {H}_{eff}[0, \mathrm{loc}_P]\mathbf{\delta}[{\mathrm{loc}_P}] \\ {H}_{eff}[1, (\mathrm{loc}_1+1)_N]\mathbf{\delta}[{(\mathrm{loc}_1+1)_N}] & \cdots & {H}_{eff}[1, (\mathrm{loc}_P+1)_N]\mathbf{\delta}_{(\mathrm{loc}_P+1)_N} \\ \vdots & \ddots & \vdots \\ {H}_{eff}[N-1, (\mathrm{loc}_1+N-1)_N]\mathbf{\delta}[{(\mathrm{loc}_1+N-1)_N}] & \cdots & {H}_{eff}[N-1, (\mathrm{loc}_P+N-1)_N]\mathbf{\delta}[{(\mathrm{loc}_P+N-1)_N}] \end{array}\right]},
\end{aligned}
\vspace{-4mm}
\end{equation}
}
we should show that
\begin{equation}\label{eq:ind_eq}
\begin{aligned}
    \beta_1\mathbf{H}_1\boldsymbol{\delta} + \beta_2\mathbf{H}_2\boldsymbol{\delta} + \cdots + \beta_P\mathbf{H}_P\boldsymbol{\delta} = \textbf{0}
    \rightarrow &\beta_1 = \beta_2 = \cdots = \beta_P = 0,
    \end{aligned}
\end{equation}
which is proved by contradiction. Assume that there is at least one $\beta_i \neq 0$ and $\beta_1\mathbf{H}_1\boldsymbol{\delta} + \beta_2\mathbf{H}_2\boldsymbol{\delta} + \cdots + \beta_P\mathbf{H}_P\boldsymbol{\delta} =  \textbf{0}$. Without loss of generality (wlog), we assume $\beta_1 \neq 0$. Dividing both sides of the vector equality in \eqref{eq:ind_eq} by $\beta_1$ and considering the first entry of the resulting vector, we have
\begin{equation}\label{eq:dela_q1}
    \delta[{\mathrm{loc}_1}] = -\frac{{H}[0, \mathrm{loc}_2)]}{{H}[0, \mathrm{loc}_1]}\frac{\beta_2}{\beta_1}\delta[{\mathrm{loc}_2}]- \cdots -\frac{{H}[0, \mathrm{loc}_P]}{{H}[0, \mathrm{loc}_1]}\frac{\beta_P}{\beta_1}\delta[{\mathrm{loc}_P]}.
\end{equation}
In addition, by taking to account the $\mathbf{H}_{eff}$ expression, we have
\begin{equation}\label{eq:H1_H2)}
    \frac{{H}[0, \mathrm{loc}_i]}{{H}[0, \mathrm{loc}_j]} = e^{\imath2\pi c_2(\mathrm{loc}_i^2-\mathrm{loc}_j^2)}e^{\imath\frac{2\pi}{N} (Nc_1(l_i^2-l_j^2) - \mathrm{loc}_il_i + \mathrm{loc}_jl_j)}.
\end{equation}
Now we can rewrite \eqref{eq:dela_q1} using \eqref{eq:H1_H2)}
\begin{equation}\label{eq:deltaq1_2}
\begin{aligned}
    \delta[{\mathrm{loc}_1}] = e^{-\imath2\pi c_2\mathrm{loc}_1^2}e^{\imath\frac{2\pi}{N}(Nc_1(-l_1^2 + \mathrm{loc}_1l_1))}
    \left[\sum_{i = 2}^{P}e^{\imath2\pi c_2\mathrm{loc}_i^2}e^{\imath\frac{2\pi}{N}(Nc_1l_i^2-\mathrm{loc}_il_i)}\beta_i'\delta[{q_i}]
    \right],
    \end{aligned}
\end{equation}
where $\beta_i' = \frac{-\beta_i}{\beta_1}$. Note that $\boldsymbol{\delta} \in \mathbb{Z}[j]^{N\times 1}$, therefore since $\delta[{\rm{loc}_1}] \in \mathbb{Z}[j]$, then the r.h.s. of \eqref{eq:deltaq1_2} should be in $\mathbb{Z}[j]$ to have the equality. On the other hand, choosing an irrational number for $c_2$, then $e^{\imath2\pi c_2q_i^2}$ is an irrational number and since \eqref{eq:deltaq1_2} should hold for different values of $\mathbf{\delta}$, the effect of $c_2$ should be removed from this equation. This can be done by choosing 
\begin{equation}\label{eq:alpha'}
    \beta_i' = e^{\imath2\pi c_2q_1^2}e^{-\imath2\pi c_2q_i^2}\mu_i, i = 2, \cdots, P
\end{equation}
where $\mu_i$s do not contain $c_2$ in their phases. Now in order to have \eqref{eq:deltaq1_2} hold, at least another $\beta_i$ should be non-zero. Again wlog, assume the non-zero one is $\beta_2$. Dividing both sides of the vector equality in \eqref{eq:ind_eq} with $\beta_2$ and considering the second entry of the resulting vector we get
{
\begin{equation}
\begin{aligned}
    \delta[{(\mathrm{loc}_2+1)_N}] = & e^{-\imath2\pi c_2(\mathrm{loc}_2+1)_N^2}e^{\imath\frac{2\pi}{N}(Nc_1(-l_2^2 + (\mathrm{loc}_2+1)_Nl_2))}\\
    &\times \left[\sum_{i = 1, i\neq2}^{P}e^{\imath2\pi c_2(\mathrm{loc}_i+1)_N^2}e^{\imath\frac{2\pi}{N}(Nc_1l_i^2-(\mathrm{loc}_i+1)_Nl_i)}\beta_i''\delta[{(\mathrm{loc}_i+1)_N}]
    \right],
\end{aligned}
\end{equation}
}
where $\beta_i'' = \frac{-\beta_i}{\beta_2}$. With the same explanation, $\beta_i'' $s are
\begin{equation}\label{eq:alpha''}
    \beta_i'' = e^{\imath2\pi c_2(q_2+1)_N^2}e^{-\imath2\pi c_2(q_i+1)_N^2}\mu'_i, i = 1, \cdots, P, i\neq 2.
\end{equation}
Putting \eqref{eq:alpha'} and \eqref{eq:alpha''} together shows
\begin{equation}
    \beta_2' = \frac{1}{\beta_1''},
\end{equation}
which leads to 
\begin{equation}\label{eq:check}
    e^{\imath2\pi c_2\left(\mathrm{loc}_2^2 - (\mathrm{loc}_2+1)_N^2 +\mathrm{loc}_1^2 - (\mathrm{loc}_1+1)_N^2 \right) }\mu_2\mu_1' = 1.
\end{equation}
On one hand, since $\mathrm{loc}_1\neq \mathrm{loc}_2$, $\mathrm{loc}_2^2 - (\mathrm{loc}_2+1)_N^2 +\mathrm{loc}_1^2 - (\mathrm{loc}_1+1)_N^2$ cannot be zero and on the other hand, we know that $\mu_i$ and $\mu_i'$ do not contain $c_2$ in their phases. Therefore the $c_2$ effect in the phase cannot be removed and the l.h.s. of \eqref{eq:check} is an irrational number while the right one is integer. Thus, our initial assumption does not hold and $\beta_1 = \cdots = \beta_P = 0$, which means that the $P$ columns of $\mathbf{\Phi}(\boldsymbol{\delta})$ are linearly independent, i.e., the rank of $\mathbf{\Phi}(\boldsymbol{\delta})$ is $P$. 
\vspace{-4mm}
\section{Proof of theorem \ref{theo:convergence}}\label{sec:appen_convproof}
In order to prove the convergence, we should show that $|\lambda(-\mathbf{S}^{-1}(\mathbf{R} - \mathbf{S}))| < 1$, where $\lambda$ denotes any eigenvalue of $-\mathbf{S}^{-1}(\mathbf{R} - \mathbf{S})$. For the corresponding eigenvector $\mathbf{v}$, where $\mathbf{v}^H\mathbf{v} = \beta>0$, we can write
\begin{equation}
    -\mathbf{S}^{-1}(\mathbf{R} - \mathbf{S})\mathbf{v} = \lambda(-\mathbf{S}^{-1}(\mathbf{R} - \mathbf{S}))\mathbf{v}.\label{eq:eig_vec1}
\end{equation}
After multiplying both sides of \eqref{eq:eig_vec1} by $\mathbf{v}^H\mathbf{S}$, it writes as
\begin{equation}
    \lambda(-\mathbf{S}^{-1}(\mathbf{R} - \mathbf{S})) = \frac{\mathbf{v}^H(-(\mathbf{R} - \mathbf{S}))\mathbf{v}}{\mathbf{v}^H\mathbf{S}\mathbf{v}}.\label{eq:eig_vec2}
\end{equation}
Considering $\mathbf{S} = (\gamma^{-1} + d)\mathbf{I} + \mathbf{L}$ and $\mathbf{R} =\underline{\mathbf H}_{\mathrm{eff}}^{H}\underline{\mathbf H}_{\mathrm{eff}}+\gamma^{-1}\mathbf{I} = \mathbf{L} + \mathbf{L}^H+(d + \gamma^{-1})\mathbf{I} $, \eqref{eq:eig_vec2} becomes
\begin{equation}
    \lambda(-\mathbf{S}^{-1}(\mathbf{R} - \mathbf{S})) =\frac{-\mathbf{v}^H\mathbf{L}^{H}\mathbf{v}}{(\gamma^{-1} + d)\mathbf{v}^H\mathbf{v} + \mathbf{v}^H\mathbf{L}\mathbf{v}}.\label{eq:eig_vec3}
\end{equation}
Since $\mathbf{R}$ is positive definite Hermitian matrix, any non-zero vector including $\mathbf{v}$ satisfies
\begin{align}
\mathbf{v}^H(\underline{\mathbf H}_{\mathrm{eff}}^H\underline{\mathbf H}_{\mathrm{eff}} + \gamma^{-1}\mathbf{I})\mathbf{v} = \mathbf{v}^H(\mathbf{L} + \mathbf{L}^H +(d + \gamma^{-1})\mathbf{I})\mathbf{v} = \beta(d + \gamma^{-1})+2\mathcal{R}(\mathbf{v}^H\mathbf{L}\mathbf{v}) >0, 
\label{eq:a_inequality}
\end{align}
where \eqref{eq:a_inequality} can be written as 
\begin{equation}
    a = \mathcal{R}(\mathbf{v}^H\mathbf{L}\mathbf{v}) = \mathcal{R}(\mathbf{v}^H\mathbf{L}^{H}\mathbf{v}) > \frac{-\beta(d + \gamma^{-1})}{2},\label{eq:ineq_2}
\end{equation}
and the imaginary part is equal to 
\begin{equation}
    b = \mathcal{I}(\mathbf{v}^H\mathbf{L}\mathbf{v}) = -\mathcal{I}(\mathbf{v}^H\mathbf{L}^{H}\mathbf{v}).
\end{equation}
Therefore, \eqref{eq:eig_vec2} can be rewritten as
\begin{equation}
    |\lambda(-\mathbf{S}^{-1}(\mathbf{R} - \mathbf{S}))| =\frac{|a-\imath b|}{|\beta(\gamma^{-1} + d)+a -\imath b |}.\label{eq:eig_vec3}
\end{equation}
From \eqref{eq:ineq_2} and \eqref{eq:eig_vec3}, it can be seen that $|\lambda(-\mathbf{S}^{-1}(\mathbf{R} - \mathbf{S}))| < 1$.

\ifCLASSOPTIONcaptionsoff
  \newpage
\fi

\bibliographystyle{IEEEtran}
\bibliography{IEEEabrv,Citations}

\begin{thebibliography}{10}
\providecommand{\url}[1]{#1}
\csname url@samestyle\endcsname
\providecommand{\newblock}{\relax}
\providecommand{\bibinfo}[2]{#2}
\providecommand{\BIBentrySTDinterwordspacing}{\spaceskip=0pt\relax}
\providecommand{\BIBentryALTinterwordstretchfactor}{4}
\providecommand{\BIBentryALTinterwordspacing}{\spaceskip=\fontdimen2\font plus
\BIBentryALTinterwordstretchfactor\fontdimen3\font minus
  \fontdimen4\font\relax}
\providecommand{\BIBforeignlanguage}[2]{{%
\expandafter\ifx\csname l@#1\endcsname\relax
\typeout{** WARNING: IEEEtran.bst: No hyphenation pattern has been}%
\typeout{** loaded for the language `#1'. Using the pattern for}%
\typeout{** the default language instead.}%
\else
\language=\csname l@#1\endcsname
\fi
#2}}
\providecommand{\BIBdecl}{\relax}
\BIBdecl

\bibitem{bemani2021afdm_ICC}
A.~Bemani, N.~Ksairi, and M.~Kountouris, ``{AFDM}: A full diversity next
  generation waveform for high mobility communications,'' in \emph{2021 IEEE
  International Conference on Communications Workshops (ICC Workshops)}, 2021,
  pp. 1--6.

\bibitem{bemani2022low}
A.~Bemani, N.~Ksairi, and M.~Kountouris, ``Low complexity equalization for afdm
  in doubly dispersive channels,'' in \emph{2022 Proc. IEEE International
  Conference on Acoustics, Speech and Signal Processing (ICASSP)}.\hskip 1em
  plus 0.5em minus 0.4em\relax IEEE, 2022.

\bibitem{wang2006performance}
T.~Wang, J.~G. Proakis, E.~Masry, and J.~R. Zeidler, ``Performance degradation
  of {OFDM} systems due to {Doppler} spreading,'' \emph{IEEE Transactions on
  Wireless Communications}, vol.~5, no.~6, pp. 1422--1432, 2006.

\bibitem{darlington1947}
``The industrial reorganization act: The communications industry,'' in
  \emph{Proceedings of the Institution of Electrical Engineers}, vol.~73,
  no.~3, 1973, pp. 635--676.

\bibitem{gott1971hf}
G.~Gott and J.~Newsome, ``{HF} data transmission using chirp signals,'' in
  \emph{Proceedings of the Institution of Electrical Engineers}, vol. 118,
  no.~9.\hskip 1em plus 0.5em minus 0.4em\relax IET, 1971, pp. 1162--1166.

\bibitem{kadri2009low}
A.~Kadri, R.~K. Rao, and J.~Jiang, ``Low-power chirp spread spectrum signals
  for wireless communication within nuclear power plants,'' \emph{Nuclear
  Technology}, vol. 166, no.~2, pp. 156--169, 2009.

\bibitem{he2010underwater}
C.~He, M.~Ran, Q.~Meng, and J.~Huang, ``Underwater acoustic communications
  using m-ary chirp-{DPSK} modulation,'' in \emph{IEEE 10th international
  conference on signal processing proceedings}.\hskip 1em plus 0.5em minus
  0.4em\relax IEEE, 2010, pp. 1544--1547.

\bibitem{palmese2008experimental}
M.~Palmese, G.~Bertolotto, A.~Pescetto, and A.~Trucco, ``Experimental
  validation of a chirp-based underwater acoustic communication method,'' in
  \emph{Proceedings of Meetings on Acoustics}, vol.~4, no.~1.\hskip 1em plus
  0.5em minus 0.4em\relax Acoustical Society of America, 2008.

\bibitem{4299496}
``{IEEE} standard for information technology - local and metropolitan area
  networks - specific requirements - part 15.4: Wireless medium access control
  {(MAC)} and physical layer {(PHY)} specifications for low-rate wireless
  personal area networks {(WPANs)}: Amendment 1: Add alternate {PHYs},''
  \emph{IEEE Std 802.15.4a-2007 (Amendment to IEEE Std 802.15.4-2006)}, pp.
  1--210, 2007.

\bibitem{martone2001multicarrier}
M.~Martone, ``A multicarrier system based on the fractional {Fourier} transform
  for time-frequency-selective channels,'' \emph{IEEE Trans. on Commun.},
  vol.~49, no.~6, pp. 1011--1020, Jun. 2001.

\bibitem{erseghe2005multicarrier}
T.~Erseghe, N.~Laurenti, and V.~Cellini, ``A multicarrier architecture based
  upon the affine {Fourier} transform,'' \emph{IEEE Trans. on Commun.},
  vol.~53, no.~5, pp. 853--862, May 2005.

\bibitem{stojanovic2010multicarrier}
D.~Stojanovi{\'c}, I.~Djurovi{\'c}, and B.~R. Vojcic, ``Multicarrier
  communications based on the affine {F}ourier transform in doubly-dispersive
  channels,'' \emph{EURASIP Journal on Wireless Communications and Networking},
  vol. 2010, pp. 1--10, 2010.

\bibitem{ouyang2016orthogonal}
X.~Ouyang and J.~Zhao, ``Orthogonal chirp division multiplexing,'' \emph{IEEE
  Trans. on Commun.}, vol.~64, no.~9, pp. 3946--3957, Sept. 2016.

\bibitem{omar2021performance}
M.~S. Omar and X.~Ma, ``Performance analysis of {OCDM} for wireless
  communications,'' \emph{IEEE Trans. on Wirel. Commun.}, vol.~20, no.~7, pp.
  4032--4043, 2021.

\bibitem{myung2006single}
H.~G. Myung, J.~Lim, and D.~J. Goodman, ``Single carrier {FDMA} for uplink
  wireless transmission,'' \emph{IEEE vehicular technology magazine}, vol.~1,
  no.~3, pp. 30--38, 2006.

\bibitem{fettweis2009gfdm}
G.~Fettweis, M.~Krondorf, and S.~Bittner, ``{GFDM}-generalized frequency
  division multiplexing,'' in \emph{VTC Spring 2009-IEEE 69th Vehicular
  Technology Conference}.\hskip 1em plus 0.5em minus 0.4em\relax IEEE, 2009,
  pp. 1--4.

\bibitem{michailow2014generalized}
N.~Michailow, M.~Matth{\'e}, I.~S. Gaspar, A.~N. Caldevilla, L.~L. Mendes,
  A.~Festag, and G.~Fettweis, ``Generalized frequency division multiplexing for
  5th generation cellular networks,'' \emph{IEEE Trans. on Commun.}, vol.~62,
  no.~9, pp. 3045--3061, 2014.

\bibitem{hadani2017orthogonal}
R.~Hadani, S.~Rakib, M.~Tsatsanis, A.~Monk, A.~J. Goldsmith, A.~F. Molisch, and
  R.~Calderbank, ``Orthogonal time frequency space modulation,'' in \emph{2017
  IEEE Wireless Communications and Networking Conference (WCNC)}, 2017.

\bibitem{hadani2018otfs}
R.~Hadani and A.~Monk, ``{OTFS}: A new generation of modulation addressing the
  challenges of 5{G},'' \emph{arXiv preprint arXiv:1802.02623}, 2018.

\bibitem{anwar2020performance}
W.~Anwar, A.~Krause, A.~Kumar, N.~Franchi, and G.~P. Fettweis, ``Performance
  analysis of various waveforms and coding schemes in {V2X} communication
  scenarios,'' in \emph{2020 IEEE Wireless Commun. \& Netw. Conf.
  (WCNC)}.\hskip 1em plus 0.5em minus 0.4em\relax IEEE, 2020, pp. 1--8.

\bibitem{raviteja2019effective}
P.~Raviteja, Y.~Hong, E.~Viterbo, and E.~Biglieri, ``Effective diversity of
  {OTFS} modulation,'' \emph{IEEE wireless communications letters}, vol.~9,
  no.~2, pp. 249--253, 2019.

\bibitem{surabhi2019diversity}
G.~Surabhi, R.~M. Augustine, and A.~Chockalingam, ``On the diversity of uncoded
  {OTFS} modulation in doubly-dispersive channels,'' \emph{IEEE Trans. on
  Wireless Communications}, vol.~18, no.~6, pp. 3049--3063, 2019.

\bibitem{raviteja2019embedded}
P.~Raviteja, K.~T. Phan, and Y.~Hong, ``Embedded pilot-aided channel estimation
  for {OTFS} in delay--{D}oppler channels,'' \emph{IEEE Trans. on Vehicular
  Technology}, vol.~68, no.~5, pp. 4906--4917, 2019.

\bibitem{bemani2021affine}
A.~Bemani, G.~Cuozzo, N.~Ksairi, and M.~Kountouris, ``Affine frequency division
  multiplexing for next-generation wireless networks,'' in \emph{2021 17th
  International Symposium on Wireless Communication Systems (ISWCS)}.\hskip 1em
  plus 0.5em minus 0.4em\relax IEEE, 2021, pp. 1--6.

\bibitem{healy2015linear}
J.~J. Healy, M.~A. Kutay, H.~M. Ozaktas, and J.~T. Sheridan, \emph{Linear
  canonical transforms: Theory and applications}.\hskip 1em plus 0.5em minus
  0.4em\relax Springer, 2015, vol. 198.

\bibitem{pei2000closed}
S.-C. Pei and J.-J. Ding, ``Closed-form discrete fractional and affine
  {Fourier} transforms,'' \emph{IEEE Trans. on Sig. Proc.}, vol.~48, no.~5, pp.
  1338--1353, May 2000.

\bibitem{bjorck1996numerical}
{\AA}.~Bj{\"o}rck, \emph{Numerical methods for least squares problems}.\hskip
  1em plus 0.5em minus 0.4em\relax SIAM, 1996.

\bibitem{saad2003iterative}
Y.~Saad, \emph{Iterative methods for sparse linear systems}.\hskip 1em plus
  0.5em minus 0.4em\relax SIAM, 2003.

\end{thebibliography}

\end{document}